\documentclass[11pt]{article}

\usepackage[utf8]{inputenc}
\usepackage[margin=1in]{geometry}
\usepackage{amsmath, amssymb, amsthm, dsfont, mathtools}
\usepackage{xcolor}
\usepackage[numbers,sort&compress]{natbib}
\usepackage{float}
\usepackage{hyperref}
\providecommand{\doi}[1]{\href{https://doi.org/#1}{doi: \nolinkurl{#1}}}
\usepackage[noabbrev, nameinlink, capitalise]{cleveref}
\usepackage{authblk}
\usepackage{enumitem}
\usepackage{aliascnt}

\usepackage{tocloft}
\usepackage{pgfplots}
\pgfplotsset{compat=1.9}
\usetikzlibrary{arrows.meta}

\definecolor{TealDark}{HTML}{008080}
\definecolor{SlateBlueDark}{HTML}{6A5ACD}
\definecolor{SubmissionPurple}{HTML}{7B3294}
\definecolor{SubmissionRed}{HTML}{B2182B}
\hypersetup{
	pdftitle={Dissipation-Enabled Operator-Norm Locality Bounds for Bose-Hubbard Hamiltonians},
	colorlinks,
	linkcolor={TealDark},
	citecolor={SlateBlueDark},
	urlcolor={SlateBlueDark}
}

\newtheoremstyle{myplain}{5pt}{5pt}{\itshape}{0pt}{\bfseries}{}{5pt plus 1pt minus 1pt}{}
\newtheoremstyle{mydefinition}{5pt}{5pt}{\normalfont}{0pt}{\bfseries}{}{5pt plus 1pt minus 1pt}{}

\theoremstyle{myplain}
\newtheorem{thm}{Theorem}[section]
\newaliascnt{conj}{thm}

\aliascntresetthe{conj}
\newaliascnt{cor}{thm}
\newtheorem{cor}[cor]{Corollary}
\aliascntresetthe{cor}
\newaliascnt{lem}{thm}
\newtheorem{lem}[lem]{Lemma}
\aliascntresetthe{lem}
\newaliascnt{prop}{thm}
\newtheorem{prop}[prop]{Proposition}
\aliascntresetthe{prop}
\newaliascnt{assum}{thm}

\aliascntresetthe{assum}

\theoremstyle{mydefinition}
\newaliascnt{defi}{thm}
\newtheorem{defi}[defi]{Definition}
\aliascntresetthe{defi}

\theoremstyle{remark}

\newtheorem{rmk}[thm]{Remark}

\crefname{thm}{theorem}{theorems}
\Crefname{thm}{Theorem}{Theorems}
\crefname{lem}{lemma}{lemmas}
\Crefname{lem}{Lemma}{Lemmas}
\crefname{cor}{corollary}{corollaries}
\Crefname{cor}{Corollary}{Corollaries}
\crefname{prop}{proposition}{propositions}
\Crefname{prop}{Proposition}{Propositions}
\crefname{defi}{definition}{definitions}
\Crefname{defi}{Definition}{Definitions}
\crefname{assum}{assumption}{assumptions}
\Crefname{assum}{Assumption}{Assumptions}
\crefname{conj}{conjecture}{conjectures}
\Crefname{conj}{Conjecture}{Conjectures}
\crefname{ex}{example}{examples}
\Crefname{ex}{Example}{Examples}
\crefname{rmk}{remark}{remarks}
\Crefname{rmk}{Remark}{Remarks}
\crefname{equation}{equation}{equations}
\Crefname{equation}{Equation}{Equations}
\crefname{section}{section}{sections}
\Crefname{section}{Section}{Sections}
\crefname{appendix}{appendix}{appendices}
\Crefname{appendix}{Appendix}{Appendices}
\crefname{figure}{figure}{figures}
\Crefname{figure}{Figure}{Figures}
\crefname{table}{table}{tables}
\Crefname{table}{Table}{Tables}

\newcommand{\N}{\mathbb{N}}
\newcommand{\cB}{\mathcal{B}}
\newcommand{\cD}{\mathcal{D}}
\newcommand{\cH}{\mathcal{H}}
\newcommand{\cI}{\mathcal{I}}
\newcommand{\cL}{\mathcal{L}}

\newcommand{\cS}{\mathcal{S}}
\newcommand{\cT}{\mathcal{T}}
\newcommand{\cW}{\mathcal{W}}
\newcommand{\id}{\mathrm{id}}
\newcommand{\1}{\mathds{1}}

\newcommand{\ketbra}[2]{\left\vert#1\right\rangle\!\! \left\langle#2\right\vert}
\newcommand{\ket}[1]{\left\vert #1 \right\rangle}
\newcommand{\bra}[1]{\left\langle #1 \right\vert}
\newcommand{\tr}[1][]{\operatorname{tr}%
  \ifx#1\empty \else \left[ #1 \right] \fi}

\title{\texorpdfstring{Dissipation-Enabled Operator-Norm\\ Locality Bounds for Bose-Hubbard Hamiltonians}{Dissipation-Enabled Operator-Norm Locality Bounds for Bose-Hubbard Hamiltonians}}
\date{}
\author[1]{Marius Lemm\thanks{marius.lemm@uni-tuebingen.de}}
\author[2,1]{Tim M\"{o}bus\thanks{Corresponding author: moebustim@gmail.com}}

\affil[1]{Department of Mathematics, University of T\"ubingen, 72076 T\"ubingen, Germany}
\affil[2]{Department of Applied Mathematics and Theoretical Physics, University of Cambridge, Cambridge CB3 0WA, United Kingdom}

\begin{document}
\hypersetup{pageanchor=false}
\maketitle
\vspace*{-6ex}
\begin{abstract}
	Closed bosonic lattice systems do not generally admit operator norm Lieb--Robinson bounds uniformly in the initial state. The reason is that the information propagation velocity can grow with the local boson occupancy, which can be macroscopically large. Here, we show that local dissipation restores an operator-norm Lieb--Robinson bound. We consider the dissipative Bose--Hubbard model described by a Lindbladian operator with on-site $\ell$-photon loss, $\ell>2$. For our main result, we establish an almost-ballistic Lieb-Robinson bound. Dissipation rapidly depletes highly occupied sites, thus regularizing the state on the Sobolev-type scale of local particle moments. The resulting moment bounds diverge as time approaches zero but remain integrable near zero at sufficiently low orders, which is crucial for our argument. We extend these ideas to treat cat-code dissipation and, for initial states in the code space, we prove a local adiabatic approximation uniform in the total volume. Further consequences include local channel approximation, a thermodynamic limit, and efficient hybrid bosonic--qubit simulation. The point is that these applications are now available \textit{uniformly in the input state}, as for quantum spin systems, but in contrast to closed Bose--Hubbard systems.
\end{abstract}

\vspace*{-1ex}
\setlength{\cftbeforesecskip}{2pt}
\setlength{\cftbeforesubsecskip}{1pt}
\tableofcontents
\vspace*{\fill}~

\thispagestyle{empty}
\newpage
\addtocounter{page}{-1}
\hypersetup{pageanchor=true}

\section{Introduction}\label{sec:intro}
    The principle of locality is a cornerstone of quantum many-body physics. For lattice systems with finite-dimensional local Hilbert spaces and short-range interactions, also known as quantum spin systems, the seminal theorem of Lieb and Robinson shows that information cannot propagate arbitrarily fast, but is effectively confined to a linear light cone with exponentially small tails outside \cite{Lieb.1972}. The Lieb--Robinson bound, later extended in \cite{hastings2006spectral,Nachtergaele.2006}, has become a fundamental tool in mathematical physics and quantum information theory, with far-reaching consequences for existence of thermodynamic limits, correlation decay and area laws for gapped ground states, and efficient simulability. See  \cite{chen2023speed} for a recent review. Analogous locality bounds are also known for broad classes of irreversible finite-dimensional dynamics generated by Lindbladians \cite{Poulin.2010,Nachtergaele.2011,Rubiliani.2026}.

	For bosonic lattice systems, however, this picture becomes substantially more delicate. The local Hilbert spaces are infinite-dimensional, and the generator is built from unbounded local terms such as the creation, annihilation, and number operators. The infinite-dimensionality of the local Hilbert space breaks the standard Lieb--Robinson bound, in which the information propagation velocity is proportional to the operator norm of the local interaction.  Hence, we cannot expect bosonic Lieb--Robinson bounds to hold uniformly over the choice of initial states, in the absence of special structure, e.g., for perturbations of harmonic oscillator models \cite{nachtergaele2007lieb}.

	A paradigmatic class of models for lattice bosons is given by \textit{Bose--Hubbard-type Hamiltonians}, which were introduced in \cite{Gersch.1963} and later became a standard effective model for ultracold bosons in optical lattices \cite{Jaksch.1998,Greiner.2002,Bloch.2008}.  Since the standard Lieb--Robinson bounds are essentially meaningless for these models and the special structure exploited in \cite{nachtergaele2007lieb} is unavailable, substantial recent literature in mathematical physics has been devoted to studying \textit{state-dependent Lieb--Robinson type bounds}. These bounds then hold for a class of ``well-behaved'' initial states and establish Lieb--Robinson type locality properties of the dynamics uniformly within this class. Examples of suitable initial states include perturbations of stationary states \cite{KuwaharaSaito.2021}, states with particle-free regions \cite{schuch2011information,Faupin.2022}, and states of controlled boson density \cite{YinLucas.2022,kuwahara2024effective,KuwaharaLemm.2024,lemm2025quantum,deuchert2025dynamics,lemm2026microscopic,lemm2026lieb}.
	
	The following free-hopping example, based on the finite-chain calculation in \cite[Supplemental Material, Section~I]{FaupinLemmSigalZhang.2025}, shows why uniformity fails. On the chain $V=\{0,\ldots,r\}$, consider the free hopping Hamiltonian $H=-\frac12\sum_{x=0}^{r-1}(a_x^\dagger a_{x+1}+a_{x+1}^\dagger a_x)$.
	Let $\tau_t$ and $\tau_t^R$ denote its full and restricted Heisenberg evolutions, where $R=\{j,\ldots,r\}$ for any integer $1\le j\le r$, and let $A_r=\1-\ketbra{0}{0}_r$ detect a nonzero occupation at $r$.
	Starting with $n$ bosons at $0$ and vacuum elsewhere (see Figure~\ref{fig:free-hopping-counterexample}), the detection probability tends to one as $n\to\infty$, whereas the restricted evolution predicts zero. Consequently,
	\begin{equation*}
		\|\tau_t(A_r)-\tau_t^R(A_r)\|_\infty=1,
		\qquad 1\le j\le r,\quad 0<t\le1.
	\end{equation*}
	The region $R$ contains every site within distance $r-j$ of the detector, yet the error is one for every cut position $j$. Thus no local approximation bound can decay with the buffer radius uniformly over all input states and finite volumes.
	\begin{figure}[H]
		\vspace{-2ex}
		\centering
		\includegraphics[width=0.84\linewidth]{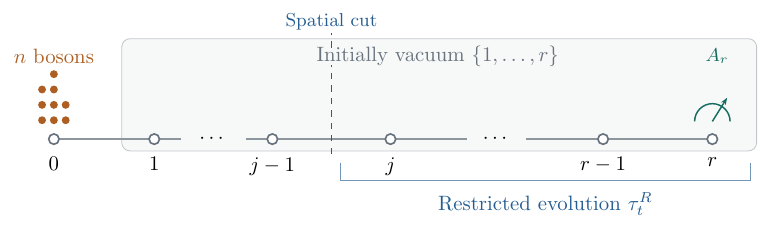}
		\vspace{-3ex}
		\caption{Failure of uniform locality under free hopping. See Appendix~\ref{app:free-hopping-obstruction} for the proof.}
		\label{fig:free-hopping-counterexample}
	\end{figure}

	Stronger counterexamples were constructed by Eisert and Gross: their one-dimensional, translationally invariant, nearest-neighbor bosonic models exhibit exponentially accelerating information propagation with an encoding independent of the transmission distance \cite{Eisert.2009}.
    Moreover, \cite{kuwahara2024effective} showed that even for bounded-density initial states,  certain Bose-Hubbard type Hamiltonians are still capable of propagating information a distance $t^D$ within time $t$, where $D$ is the spatial dimension. (Step 2 of the construction by \cite{kuwahara2024effective} is general and applies at any local particle density, also a macroscopically large one, and thus gives examples of bosonic lattice Hamiltonians with macroscopically large information propagation speed arising for general initial states.) To summarize, closed Bose--Hubbard dynamics only permit state-dependent Lieb--Robinson type bounds and such bounds are often super-ballistic (but uniform in system size).

    In this work, we study how on-site dissipation can be used to improve the picture for the dissipative Bose--Hubbard model. Dissipative Bose--Hubbard dynamics appears naturally when lattice bosons are subject to localized one-body losses, controlled two-body losses, or engineered reservoirs \cite{Syassen.2008,GarciaRipoll.2009,Kepesidis.2012,Kordas.2015,Tomita.2019}. It also belongs to the broader program of using dissipation as a resource for preparing and controlling many-body quantum systems \cite{Diehl.2008,Mueller.2012}. We consider a bosonic lattice model with local photon loss, described on the bosonic Fock space $\cH$ over a lattice $V$ with edges $E$ by the Lindblad master equation
	\begin{equation}\label{eq:lindblad}
        \frac{d\rho_t}{dt} = \cL(\rho_t) = -i[H,\rho_t] + \sum_{j \in V} \gamma_j \cL[a_j^\ell](\rho_t),
	\end{equation}
    where $\cL[L]=L\cdot L^\dagger-\frac{1}{2}\{L^\dagger L,\cdot\}$. Throughout, the loss rates are uniformly bounded away from zero and infinity,
	\begin{equation*}
		0<\gamma_{\min}\coloneqq\inf_{j\in V}\gamma_j
		\le \gamma_j\le
		\gamma_{\max}\coloneqq\sup_{j\in V}\gamma_j<\infty,
	\end{equation*}
	with constants independent of the finite volume. The Hamiltonian $H$ is the Bose--Hubbard Hamiltonian
	\begin{equation}\label{eq:bose-hubbard}
        H = -\sum_{(i,j)\in E} (J_{ij}a_i^\dagger a_j + \overline{J_{ij}}a_j^\dagger a_i) + \frac{1}{2} \sum_{i\in V} U_i N_i(N_i - 1)-\sum_{i\in V}\mu_i N_i.
	\end{equation}
	Here $\ell\in\mathbb N$ and $U_i,\mu_i\in\mathbb R$. We write $e=(i,j)$ for an undirected edge with endpoints $i,j$, identifying $(i,j)$ and $(j,i)$. Sums over edges count each edge once. When an ordered hopping matrix is used in estimates, we extend the coefficients by $J_{ji}=\overline{J_{ij}}$ and set $J_{ij}=0$ if $(i,j)\notin E$.
	Here $\rho$ denotes a density operator in the domain $\cD(\cL)$ of the, in general unbounded, generator $\cL$. The guiding idea is that local loss suppresses large on-site occupations and can therefore act as a dynamical regulator for bosonic transport. Dissipative bosonic transport has also been studied from a macroscopic optimal-transport perspective for long-range systems \cite{Kuwahara.2025}. Moment-controlled locality and simulation bounds for bosonic systems, including engineered loss for polynomial interactions, were developed in \cite{Moebus.2024Learning}. Here the predual dynamics regularizes instantaneously into weighted Sobolev spaces, yielding uniform local approximation without assumptions on the initial moments. Recent progress on functional inequalities and dissipative bosonic systems supports this perspective by establishing moment-propagation estimates \cite{Gondolf.2024,Gondolf.2025thesis,Moebus.2025}. In particular, the Sobolev-regularization result of one of the authors is the analytic tool that turns local loss into a uniform estimate \cite{Rico.2025}.

	The main result of this paper is an operator-norm local approximation estimate for Equation~\eqref{eq:lindblad} when $\ell>2$. Theorem~\ref{thm:adaptive-moment-locality} compares the full Heisenberg evolution with the physical evolution on a finite buffer, uniformly in the initial state and the finite volume. 
      Theorem \ref{thm:adaptive-moment-locality} shows that in the presence of suitably strong multiphoton loss ($\ell>2$), the landmark result of \cite{kuwahara2024effective} for closed Bose--Hubbard systems can be improved in two key ways:
    \begin{itemize}
        \item Theorem \ref{thm:adaptive-moment-locality}  gives a bosonic Lieb--Robinson bound that is uniform in the choice of initial state. In particular, no bounded-density assumption on the initial state as in \cite{kuwahara2024effective} is required.
        \item The bound gives an almost-linear light cone $r\gtrsim t^{1+\eta}$ for any $\eta>0$. This should be compared with $r\gtrsim t^D$ in \cite{kuwahara2024effective}  which is super-ballistic in $D>1$ and shown there to be sharp for suitable time-dependent Hamiltonians.
    \end{itemize}
    The proof uses a time dependent occupation cutoff and optimizes the cutoff scale and moment order.  On $\mathbb Z^D$, it gives superpolynomial decay $e^{-c[\log(2+r)]^2}$ in the buffer radius $r$ at fixed time in every fixed dimension. The restriction $\ell>2$ ensures that the integrated cutoff is finite despite the regularization singularity at time zero.\\[2ex]
    The proof has three ingredients:
    \begin{enumerate}[label=(\roman*), leftmargin=2.2em]
        \item We replace the unbounded hopping terms by locally particle-number truncated hopping terms with a time dependent cutoff.
        \item The truncated dynamics satisfies a dissipative Lieb--Robinson bound controlled by the time integral of the bounded interaction strength. The on-site terms enter through contractive, support-preserving semigroups.
        \item The particle-number truncation error is estimated in the predual picture by Sobolev regularization, which is the point at which dissipation enters: local loss supplies the high-moment bounds needed for a uniform, state-independent estimate.
    \end{enumerate}
    Then, Proposition~\ref{prop:cat-code-polynomial-lrb} extends the adaptive bound to quartic polynomial interactions with shifted four-photon loss. Its constants are uniform in the common dissipation strength $\gamma\ge1$. For initial states in the code space, Theorem~\ref{thm:local-cat-adiabatic-estimate} gives a local adiabatic error $\mathcal O(\gamma^{-1+\eta})$ for every $0<\eta<1$, uniformly in the volume. Extending the present state-independent locality argument to two-photon loss is an interesting direction for future work.

	The paper is organized as follows. Section~\ref{sec:prelim} fixes the notation and functional-analytic setting. Section~\ref{sec:uniform-lrb} states the physical-buffer and polynomial cat-code bounds. Section~\ref{sec:main-proofs} collects their input results and proofs. Section~\ref{sec:appl} gives local channel approximation, Hamiltonian stability, correlation spreading and preparation-time bounds, a local cat-code adiabatic limit, and finite-register simulation. The appendices give the detailed Sobolev and bounded-interaction arguments and the finite-region adiabatic estimate, together with the free-hopping counterexample.

\section{Preliminaries}\label{sec:prelim}
	We establish the notation and recall the basic mathematical framework for bosonic quantum systems. The \emph{bosonic Fock space} is the separable Hilbert space $\cH$, with orthonormal basis $\{\ket{n} : n \in \mathbb{N}_0\}$ indexed by the photon number. Each basis vector $\ket{n}$ represents a state with exactly $n$ photons. On this space, the \emph{annihilation} and \emph{creation operators} $a$ and $a^\dagger$ act by $a\ket{n} = \sqrt{n}\ket{n-1}$ and $a^\dagger\ket{n} = \sqrt{n+1}\ket{n+1}$, satisfying the canonical commutation relation $[a, a^\dagger] = \1$. The \emph{number operator} is defined by $N \coloneqq a^\dagger a$. For multimode systems with $m$ bosonic modes indexed by $i \in V$ for a set of vertices $|V| = m$, the Fock space becomes $\cH_m = \bigotimes_{i \in V} \cH_i$ or $\cH$ if clear from the context, and for each mode $i$ we denote the annihilation, creation, and number operators as $a_i$, $a_i^\dagger$, and $N_i = a_i^\dagger a_i$, respectively. The total number operator is $N_{\mathrm{tot}}\coloneqq\sum_{i\in V}N_i$.

	For multimode systems, we work on finite spatial volumes $V\subset \mathbb{Z}^D$, with constants chosen uniformly in $V$, and an undirected edge set $E$ describing couplings. We take the finite volume $V$ connected, so the concentration weights below are strictly positive. Disconnected full models can be treated componentwise. Buffers $R\subseteq V$ need not be connected. The graph is equipped with the graph distance $\mathrm{dist}(i,j)$ measuring the shortest path between $i,j\in V$. For a region $R \subseteq V$, its vertex boundary is $\partial R = \{i \in V \setminus R : \exists j \in R \text{ with } (i,j) \in E\}$ and the $k^\text{th}$-neighborhood is $R[k] = \{i \in V : \mathrm{dist}(i, R) \leq k\}$. Infinite-volume statements should be understood through these uniform finite-volume estimates.

	The space of \emph{trace-class operators} is denoted $\cT_1(\cH)$ with trace norm $\|\rho\|_1 = \tr[|\rho|]$ and the subset of density operators (quantum states) $\cS = \{\rho \in \cT_1(\cH) : \rho \geq 0, \tr[\rho]=1\}$. Its dual is identified with the bounded operators $\cT_\infty(\cH)=\cB(\cH)$, equipped with the operator norm. For unbounded operators, the crucial domains are the space of \emph{finite vectors}. In one mode,
	\begin{equation*}
		\cH_f \coloneqq \mathrm{span}\{\ket{n} : n \in \mathbb{N}_0, \, n \le M \text{ for some } M \in \mathbb{N}\} \subset \cH,
	\end{equation*}
	while in a finite volume $V$, we use multi-indices $\mathbf n=(n_i)_{i\in V}$, write $\ket{\mathbf n}=\bigotimes_{i\in V}\ket{n_i}$, and set
	\begin{equation*}
		\cH_{V,f}
		\coloneqq
		\mathrm{span}\{\ket{\mathbf n}:\mathbf n\in\mathbb N_0^V\},
		\qquad
		\cT_{V,f}
		\coloneqq
		\mathrm{span}\{\ketbra{\mathbf n}{\mathbf m}:\mathbf n,\mathbf m\in\mathbb N_0^V\}.
	\end{equation*}
	Thus $\cT_{V,f}$ is the algebraic span of matrix units between tensor-product number vectors and is dense in $\cT_1(\cH_V)$. We write $\cH_f$ and $\cT_f$ when the volume is clear from context. All polynomial Hamiltonians below are symmetric on $\cH_f$. A bounded superoperator $K_Y$ is supported in $Y$ if it preserves every local algebra containing $Y$ and vanishes on observables supported in regions disjoint from $Y$.

	A \emph{quantum Markov semigroup} (QMS) is a strongly continuous family $(e^{t\cL})_{t \geq 0}$ of completely positive, trace-preserving contractions on $\cT_1$ satisfying $e^{t\cL} \circ e^{s\cL} = e^{(t+s)\cL}$ and $e^{0 \cdot \cL} = \1$. Its Heisenberg evolution is the normal weak* continuous adjoint contraction semigroup on $\cT_\infty=\cB(\cH)$, defined by the duality relation
	\begin{equation*}
		\tr[\tau_t(A)\rho]=\tr[Ae^{t\cL}(\rho)],\qquad A\in\cB(\cH),\ \rho\in\cT_1.
	\end{equation*}
	Thus $\tau_t$ should not be read as $\exp(t\cL^\dagger)$ acting as a $C_0$-semigroup on all of $\cB(\cH)$ in operator norm. Together with trace-norm continuity of the predual, complete positivity and unitality imply that $t\mapsto\tau_t(A)$ is strong-* continuous for each $A\in\cB(\cH)$ \cite{FagnolaGirotti.2025}. For fixed $A$, this map need not be continuous in the operator norm of $\cB(\cH)$. In the finite-volume and locally truncated arguments below we use the same dual definition. The adjoint evolutions are contractive and generated by local terms.

	For multimode systems on a lattice, we introduce \emph{quantum Sobolev spaces} with locally concentrated Sobolev norms following \cite{Rico.2025}. The reference operator $\cW_v^{k}(x)$ centered at mode $v \in V$ with decay constant $\kappa > 0$ is defined by the closure of the operator on $\cT_f$ given by
	\begin{equation}\label{eq:multi-mode-sobolev-operator}
		\cW_v^{k}(x)=M_{v,k}^{1/2}xM_{v,k}^{1/2},
		\qquad M_{v,k}=\sum_{i\in V}\frac{e^{-\kappa\,\mathrm{dist}(v,i)}}{Z_v}(N_i+\1)^k,
	\end{equation}
	where $Z_v=\sum_{j \in V} e^{-\kappa \, \mathrm{dist}(v,j)}$ is the normalization constant. For finite nearest-neighbor subgraphs of $\mathbb Z^D$, $Z_\kappa^*\coloneqq\sup_{V,v}Z_v<\infty$ for fixed $\kappa>0$. The Sobolev space centered at $v$ is
	\begin{equation}\label{eq:sobolev-norm}
		W^{k,1}_v \coloneqq \cD(\cW_v^{k}) = \left\{ x \in \cT_1 : \|x\|_{W^{k,1}_v} < \infty \right\},
	\end{equation}
	equipped with the norm $\|x\|_{W^{k,1}_v} = \| \cW_v^{k}(x) \|_1$, which forms a Banach space by the isometric sandwich identification with $\cT_1$, with inverse $y\mapsto M_{v,k}^{-1/2}yM_{v,k}^{-1/2}$ (see \cite[Lem.~D.1]{Rico.2025}). The unbounded sandwich is understood as a form with a trace-class extension. For positive $\rho$,
	\begin{equation*}
		\|\rho\|_{W_v^{k,1}}=\operatorname{tr}(M_{v,k}\rho)
		=\sum_{i\in V}\frac{e^{-\kappa\,\mathrm{dist}(v,i)}}{Z_v}
		  \operatorname{tr}[(N_i+\1)^k\rho].
	\end{equation*}
	Thus the positive moment calculations below use precisely these spaces \cite{Rico.2025}.

	The spaces introduced in \cite{Rico.2025} are an approximately local version of the single-mode Sobolev space discussed in \cite{Gondolf.2024}. The present version allows to treat many-mode systems uniformly in the total particle number, which is essential for us.

	\begin{defi}[Multi-mode Sobolev-preserving QMS]\label{def:multimode-sobolev-preserving-qms}
		Let $k>0$ and $v\in V$. We call a quantum Markov semigroup $(e^{t\cL})_{t\ge0}$ \emph{$(k,v)$-Sobolev preserving} if $e^{t\cL}$ defines a $C_0$-semigroup on $W^{k,1}_v$ and there exist constants $C,\omega>0$ such that
	\begin{equation*}
			\|e^{t\cL}\|_{W^{k,1}_v\to W^{k,1}_v}\le C e^{\omega t},\qquad t\ge0.
	\end{equation*}
	\end{defi}

	\begin{defi}[Multi-mode Sobolev-regularizing QMS]\label{def:multimode-sobolev-regularizing-qms}
		Let $k>0$ and $v\in V$. We call a QMS $(e^{t\cL})_{t\ge0}$ \emph{$(k,v)$-Sobolev regularizing} if $e^{t\cL}$ is a bounded map from $\cT_1$ to $W_v^{k,1}$ for every $t>0$, i.e.,
	\begin{equation*}
			\|e^{t\cL}\|_{1\to W^{k,1}_v}<\infty,\qquad t>0.
	\end{equation*}
	\end{defi}

\section{Main Results}\label{sec:uniform-lrb} 
	Lieb--Robinson bounds quantify the propagation of information through a lattice by limiting the influence of distant sites on a local observable. Here we bound that influence by comparing the full evolution with the evolution on a finite region containing the observable. We give bounds for the dissipative Bose--Hubbard model and polynomial cat-code dynamics. Their proofs appear in Section~\ref{sec:main-proofs}.

\subsection{Uniform bound for the Bose--Hubbard dynamics}\label{subsec:main-result}\label{subsec:adaptive-moment-cutoffs}
		For $R\subseteq V$, let $\cL_R$ be the generator obtained by retaining the on-site terms in $R$ and the hopping terms on $E(R)=\{e\in E:e\subset R\}$. Write $\tau_t$ and $\tau_t^R$ for the full and restricted Heisenberg evolutions, with $\tau_t^R$ extended by the identity outside $R$.
		Set $E_R=\{e\in E:e\cap R\ne\varnothing\}$ and $\partial_E R=\{e\in E:|e\cap R|=1\}$. For a nonempty vertex set $X\subseteq V$ and an edge set $F\subseteq E$, define
		\begin{equation*}
			d(X,F)=\min\{\mathrm{dist}(x,z):x\in X,\ z\in e,\ e\in F\},
		\end{equation*}
		with $d(X,\varnothing)=+\infty$.

		\begin{thm}[Uniform-topology dissipative Lieb--Robinson bound]\label{thm:adaptive-moment-locality}
			Let $(\cL,\cT_f)$ be the dissipative Bose--Hubbard model \eqref{eq:lindblad} on a finite graph $(V,E)$ of maximum degree at most $\Delta_G$, with $\ell>2$, $J_*=\sup_{(i,j)\in E}|J_{ij}|<\infty$, loss rates $0<\gamma_{\min}\le\gamma_j\le\gamma_{\max}<\infty$, and Sobolev weights satisfying $Z_\kappa^*<\infty$. We set $J_*=0$ when $E=\varnothing$.
			There are constants $C,c>0$, depending only on $\Delta_G$, $\ell$, $J_*$, the loss-rate bounds, and the fixed Sobolev weights, such that, for every $T\ge0$, nonempty $X\subseteq R\subseteq V$, and $A\in\cB(\cH_X)$,
			\begin{equation}\label{eq:adaptive-graph-locality}
			\begin{aligned}
				\sup_{0\le t\le T}\|\tau_t(A)-\tau_t^R(A)\|_\infty
				&\le C|E_R|(1+T)\|A\|_\infty\\
				&\qquad\times\exp\!\left[-c\log^2\!\left(2+\frac{d(X,\partial_E R)}{1+T}\right)\right].
			\end{aligned}
			\end{equation}
			On nearest-neighbor subgraphs of $\mathbb Z^D$, for every integer $r\ge0$, one also has
			\begin{equation}\label{eq:adaptive-log-gaussian-locality}
				\sup_{0\le t\le T}\|\tau_t(A)-\tau_t^{X[r]}(A)\|_\infty
				\le C|X|(1+T)\|A\|_\infty
				\exp\!\left[-c\log^2\!\left(2+\frac{r}{1+T}\right)\right],
			\end{equation}
			where $C,c$ may additionally depend on $D$. All constants are independent of $T$, $X$, $R$, $r$, the finite volume, and the on-site coefficients $U_i,\mu_i$. If $J_*=0$ or $d(X,\partial_E R)=+\infty$, then $\tau_t(A)=\tau_t^R(A)$; the right side of \eqref{eq:adaptive-graph-locality} is interpreted as zero when the distance is infinite.
		\end{thm}

	The most important point is that the bound involves the operator norm and is thus uniform in the choice of the initial state. 

    A second point concerns the shape of the light cone, which is \textit{almost ballistic} in our case.
         For fixed support and error tolerance, Equation~\eqref{eq:adaptive-log-gaussian-locality} permits a buffer radius $r=\lceil C_\eta(1+T)^{1+\eta}\rceil$ for every $\eta>0$, with a constant depending on $\eta$, the support, the tolerance, and the model bounds. An estimate where $r$ is linear in $t$ would correspond to a finite-velocity cone as in the standard Lieb-Robinson bound for quantum spin systems. In our case, the presence of the arbitrarily small $\eta>0$ amounts to an \textit{almost-ballistic propagation bound}. The constants also depend on the loss order $\ell$. In particular, the moment estimates used in the proof cannot be uniform in $\ell$: on one site, $|\ell-1\rangle\langle\ell-1|$ is stationary under $\ell$-photon loss and a Hamiltonian depending on $N$, and its expectation of $N+\1$ is $\ell$. 

    The two-photon loss case $\ell=2$ remains open. It regularizes all number moments by Proposition~\ref{main-prop:sobolev-regularizing}, but the resulting cutoff scales like $s^{-1}$ near time zero. Its time integral therefore diverges, and the present adaptive-cutoff argument does not establish locality uniformly over all input states.

\subsection{Multiphoton cat-code dissipators and polynomial interactions}\label{subsec:cat-code-extension}
	We consider shifted four-photon jumps and polynomial interactions. On a finite volume $V$, set
	\begin{equation}\label{eq:cat-code-jump}
	\begin{aligned}
		L_x&=a_x^4-\alpha_x^4,&\qquad
		H^{\mathrm{cat}}&=\sum_{x\in V}H_x+\sum_{e\in E}H_e,\\
		\cL^\gamma&=-i[H^{\mathrm{cat}},\cdot]
			+\gamma\sum_{x\in V}\gamma_x^{\mathrm{cat}}\cL[L_x],
			&\qquad\gamma&\ge1.
	\end{aligned}
	\end{equation}
	For $R\subseteq V$, define $H_R^{\mathrm{cat}}$ and $\cL_R^\gamma$ by restricting these sums to $x\in R$ and $e\in E(R)$.
	Each $H_e$ is a Hermitian polynomial of total degree at most four, where degree counts creation and annihilation operators. For $e=(x,y)$, this includes beam-splitter terms $a_x^\dagger a_y+a_y^\dagger a_x$, two-mode squeezing terms $a_xa_y+a_x^\dagger a_y^\dagger$, cross-Kerr terms $N_xN_y$, and pair-exchange terms $(a_x^\dagger)^2a_y^2+(a_y^\dagger)^2a_x^2$. Let $J_{\mathrm{cat}}$ bound the sum of absolute coefficients in each edge polynomial $H_e$. It does not include the coefficients of the on-site polynomials $H_x$, which are assumed to satisfy a separate uniform bound. Each $H_x$ is a Hermitian polynomial whose terms either commute with $N_x$ or have degree at most four. Write $\tau_t^\gamma$ for the full Heisenberg evolution and $\tau_t^{\gamma,R}$ for its restriction to $R$, extended by the identity outside $R$.

	\begin{prop}[Four-photon polynomial Lieb--Robinson bound]\label{prop:cat-code-polynomial-lrb}
		Let $(\cL^\gamma,\cT_f)$ be the polynomial bosonic model with shifted four-photon loss defined in Equation~\eqref{eq:cat-code-jump} on a lattice $(V,E)$ of uniformly bounded maximum degree $\Delta_G$, with weight satisfying $Z_\kappa^*<\infty$ and uniform bounds on the amplitudes, Hamiltonian coefficients, and rates $0<\gamma_-\le\gamma_x^{\mathrm{cat}}\le\gamma_+<\infty$. For every $p>4$, there are $C_p,C>0$, depending only on $p$ and these fixed bounds, such that, for $\lambda\ge2$, $t>0$, $X\subset R\subseteq V$, $A\in\cB(\cH_X)$, and $\gamma\ge1$,
		\begin{equation}\label{eq:cat-code-lrb}
		\begin{aligned}
			\|\tau_t^\gamma(A)-\tau_t^{\gamma,R}(A)\|_\infty
				&\le C_pJ_{\mathrm{cat}}|E_R|\,\|A\|_\infty
					\lambda^{-(p-4)/2}(t^{1/3}+t)\\
				&\quad+C|\partial_E R|\,\|A\|_\infty e^{-d(X,\partial_E R)}
					\left(e^{C_pJ_{\mathrm{cat}}(\lambda+1)^2(t^{1/3}+t)}-1\right).
		\end{aligned}
		\end{equation}
		In particular, on $\mathbb Z^D$, for every $X$, $T>0$, and $q>0$, there is $C_{q,T,X}<\infty$ such that, for $R=X[r]$,
		\begin{equation}\label{eq:cat-locality-lattice}
			\sup_{0\le t\le T}
			\|\tau_t^\gamma(A)-\tau_t^{\gamma,R}(A)\|_\infty
				\le C_{q,T,X}(1+r)^{-q}\|A\|_\infty.
		\end{equation}
		The constants are independent of the volume, $r$, and the common dissipation strength $\gamma\ge1$. They can be chosen with $C_p\le Ce^{Cp}$, where $C$ is independent of $p$.
	\end{prop}
	\section{Proof of main results}\label{sec:main-proofs}
\addtocontents{toc}{\protect\setcounter{tocdepth}{1}}
We collect the two analytic inputs and the truncation estimate before proving the locality bounds.

For the proofs, we truncate the hopping terms by particle number. For an integer $M\ge2$ and $e=(x,y)\in E$, set
\begin{align*}
	P_z^{(M)}&=\sum_{n=0}^M\ketbra{n}{n},
	& P_e^{(M)}&=P_x^{(M)}P_y^{(M)},\\
	H_e&=-J_{xy}a_x^\dagger a_y-\overline{J_{xy}}a_y^\dagger a_x,
	& H_e^{(M)}&=P_e^{(M)}H_eP_e^{(M)}.
\end{align*}
Here $\cI_e^{(M)}(\rho)=-i[H_e^{(M)},\rho]$ and $\cL_x^{\mathrm{on}}$ denotes the on-site summand of $\cL$. The truncated generator on $R\subseteq V$ is
\begin{equation}\label{eq:local-truncated-generator}
	\cL_{R,M}=\sum_{x\in R}\cL_x^{\mathrm{on}}+\sum_{e\in E(R)}\cI_e^{(M)}.
\end{equation}

\subsection{Sobolev regularization input}\label{subsec:sobolev-regularization}
		We use the following Sobolev regularization estimate for the dissipative Bose--Hubbard model, based on \cite[Cor.~3.8 and Thm.~3.12]{Rico.2025}.
		\begin{prop}[Dissipative Bose--Hubbard]\label{main-prop:sobolev-regularizing}
			Let $(\cL,\cT_f)$ be the dissipative Bose--Hubbard model \eqref{eq:lindblad} on a lattice $(V,E)$ of uniformly bounded maximum degree $\Delta_G$ and hopping coefficients $J_*=\sup_{(i,j)\in E}|J_{ij}|<\infty$ for any $\ell\ge2$ and loss rates $0<\gamma_{\min}\le\gamma_j\le\gamma_{\max}<\infty$. Then the closure of $(\cL,\cT_f)$ generates a QMS on $\cT_1$, which is Sobolev preserving and regularizing in the senses of Definitions~\ref{def:multimode-sobolev-preserving-qms} and~\ref{def:multimode-sobolev-regularizing-qms}. For every $k\ge1$, $v\in V$, and $t>0$,
			\begin{equation*}
				\|e^{t\cL}\|_{1\to W_v^{k,1}}\le\left(\frac{2k}{(\ell-1)\gamma_{\min}\ell t}\right)^{k/(\ell-1)}+\Delta_{\ell,k}^k,
			\end{equation*}
			where
			\begin{equation*}
				\Delta_{\ell,k}=\frac{2}{\gamma_{\min}\ell}\left(\gamma_{\max}\frac{(\ell+1)\ell^2}{2}+2^kkJ_*\Delta_G(1+e^\kappa)\right).
			\end{equation*}
			The constants are independent of the finite volume and center $v$.
		\end{prop}
		\begin{proof}
			The proof is given in Appendix~\ref{sec:sobolev}.
		\end{proof}

	\subsection{Bounded-interaction Lieb--Robinson input}\label{subsec:lieb-robinson-bounds}
		We first consider bounded edge interactions and unbounded on-site terms. Write the generator as
		\begin{equation}\label{eq:onsite-plus-interaction-generator}
			\cL^{(M)}=\cL_{\mathrm{on}}+\sum_{e\in E}\cI^{(M)}_e,
			\qquad
			\cI^{(M)}_e(\rho)=-i[H^{(M)}_e,\rho],
		\end{equation}
		where $\cL_{\mathrm{on}}$ is the sum of all on-site Hamiltonian and dissipative terms. In the examples considered here, the concentrated Sobolev regularization condition established in Propositions~\ref{main-prop:sobolev-regularizing} and \ref{prop:cat-code-polynomial-lrb} remains valid after removing all edges. Its Heisenberg adjoint is therefore a contractive weak* adjoint semigroup on $\cT_\infty$ and preserves supports because the on-site maps factor over sites. The terms $H^{(M)}_e$ are arbitrary but bounded self-adjoint two-site Hamiltonians. In the application below, these are the truncated Bose--Hubbard hopping terms. On every finite volume, $B_M\coloneqq \sum_{e\in E}\cI_e^{(M)}$ is bounded. Hence \cite[Chap.~III, Thm.~1.3]{Engel.2000} gives generation of a $C_0$-semigroup on $\cT_1$ by $\cL_{\mathrm{on}}+B_M$. In our Hamiltonian-interaction case the perturbations $\cI_e^{(M)}=-i[H_e^{(M)},\cdot]$ generate isometric groups, and the Trotter product formula \cite[Cor.~III.5.8]{Engel.2000} preserves complete positivity, trace preservation, and contractivity of the resulting QMS and of its weak* Heisenberg adjoint. The space $\cT_f$ is a generator core for the on-site dynamics by the same generation result and remains a core under the bounded perturbation. This includes disconnected regions.

		Under these assumptions, the following Lieb--Robinson bounds proven in \cite{Rubiliani.2026} hold (see Section~\ref{sec:lrb} for details). The key point is that the Lieb--Robinson velocity is controlled by the bounded interaction part and is insensitive to arbitrary on-site terms.

		\begin{prop}[LR bound with unbounded on-site terms]\label{prop:bounded-interaction-lrb}
			Let $A$ be a bounded operator supported in $X\subset V$ and let $\tau_t^{(M)}$ be the Heisenberg adjoint semigroup generated, in the duality sense, by $\cL^{(M)}$. There is a constant $v_{\mathrm{LR}}>0$, depending only on the graph degree, such that every bounded normal superoperator $K_Y$ supported in $Y\subset V$ with $X\cap Y=\varnothing$ satisfies
			\begin{equation}\label{eq:bounded-interaction-lrb}
				\|K_Y\tau_t^{(M)}(A)\|_\infty
				\le \frac{1}{2}\|K_Y\|_{\infty\to\infty}\|A\|_\infty \min\{|X|,|Y|\}\,
				e^{-d(X,Y)}\bigl(e^{v_{\mathrm{LR}} J_M t}-1\bigr),
			\end{equation}
			where $J_M=\sup_{e\in E}\|\cI_e^{(M)}\|_{\infty\to\infty}$.
		\end{prop}
		\begin{proof}
			The on-site cutoff hypothesis is verified in Appendix~\ref{sec:lrb}. Corollary~\ref{cor:lrb-bounded-interaction-graph} applies with $\|\Phi\|_1\le\Delta_GJ_M$ and gives Equation~\eqref{eq:bounded-interaction-lrb} with $v_{\mathrm{LR}}=2\mathrm e\Delta_G$. We retain this symbol after enlarging it by universal factors when using the concrete bounds $J_*M$ or $J_{\mathrm{cat}}M^2$ in place of $J_M$.
		\end{proof}

	\subsection{Moment-based truncation input}\label{subsec:hopping-truncation-input}
\noindent\emph{Tail estimate for the hopping truncation.}
		The following predual estimate controls the truncation error in the proofs below. Sobolev regularization supplies the required moments.
		\begin{lem}[Hopping truncation error]\label{lem:hopping-truncation-error}
			Let $p>2$ and $M\ge2$. For every positive trace-class operator $\sigma$ with $\tr\sigma\le1$ and finite local $p$-moments,
			\begin{equation}\label{eq:hopping-tail-bound}
				\|(H_e-H_e^{(M)})\sigma\|_1+\|\sigma(H_e-H_e^{(M)})\|_1
				\le C_p |J_{xy}|\,M^{-\frac{p-2}{2}}
				\sum_{z\in e}\tr[(N_z+\1)^p\sigma]^{1/2},
			\end{equation}
			where $C_p=4\left(\frac{(p-1)^{(p-1)/p}}{p}\right)^{1/2}\leq 4$.
		\end{lem}
		\begin{proof}
			Set $h_{xy}=a_x^\dagger a_y$. Since
			\begin{equation*}
				P_e^{(M)}h_{xy}P_e^{(M)}
				=h_{xy}P_x^{(M-1)}P_y^{(M)},
			\end{equation*}
			we have
			\begin{equation*}
				h_{xy}-P_e^{(M)}h_{xy}P_e^{(M)}
				=h_{xy}R_{xy}^{(M)},\qquad
				R_{xy}^{(M)}=\1-P_x^{(M-1)}P_y^{(M)}.
			\end{equation*}
			The operator $h_{xy}(N_x+\1)^{-1/2}(N_y+\1)^{-1/2}$ is a contraction. Hence, Cauchy--Schwarz inequality and $\tr\sigma\le1$ show
			\begin{equation*}
			\begin{aligned}
				\|h_{xy}R_{xy}^{(M)}\sigma\|_1
				&\le
				\tr\!\left[R_{xy}^{(M)}(N_x+\1)(N_y+\1)R_{xy}^{(M)}\sigma\right]^{1/2}.
			\end{aligned}
			\end{equation*}
			On the range of $R_{xy}^{(M)}$, either $N_x+\1\ge M$ or $N_y+\1\ge M$. Therefore, with
			\begin{equation*}
				c_p=\sup_{\max\{r,s\}\ge1}\frac{rs}{r^p+s^p}
				=\frac{(p-1)^{(p-1)/p}}{p},
			\end{equation*}
			if $\max\{r,s\}\ge M$, rescaling $r=Mu$ and $s=Mv$ gives
			\begin{equation*}
				rs=M^2uv\le c_pM^2(u^p+v^p)
				=c_pM^{-(p-2)}(r^p+s^p).
			\end{equation*}
			Functional calculus for the commuting number operators therefore gives
			\begin{equation*}
				R_{xy}^{(M)}(N_x+\1)(N_y+\1)R_{xy}^{(M)}
				\le c_pM^{-(p-2)}
				R_{xy}^{(M)}\bigl((N_x+\1)^p+(N_y+\1)^p\bigr)R_{xy}^{(M)}.
			\end{equation*}
			It follows that
			\begin{equation*}
				\|h_{xy}R_{xy}^{(M)}\sigma\|_1
				\le
				c_p^{1/2}M^{-\frac{p-2}{2}}
				\sum_{z\in e}\tr[(N_z+\1)^p\sigma]^{1/2}.
			\end{equation*}
			The same estimate holds for $h_{yx}=a_y^\dagger a_x$, with $x$ and $y$ interchanged. Thus
			\begin{equation*}
				\|(H_e-H_e^{(M)})\sigma\|_1
				\le
				2c_p^{1/2}|J_{xy}|M^{-\frac{p-2}{2}}
				\sum_{z\in e}\tr[(N_z+\1)^p\sigma]^{1/2}.
			\end{equation*}
			Since $H_e-H_e^{(M)}$ is self-adjoint and $\sigma$ is positive,
			\begin{equation*}
				\|\sigma(H_e-H_e^{(M)})\|_1
				=\|(H_e-H_e^{(M)})\sigma\|_1.
			\end{equation*}
			Adding the two bounds gives Equation~\eqref{eq:hopping-tail-bound} with $C_p=4c_p^{1/2}$.
		\end{proof}

	\subsection{Proofs of the Bose--Hubbard bounds}
\noindent Figure~\ref{fig:adaptive-locality-proof} summarizes the two comparisons and the adaptive cutoff used in the proof.
\begin{proof}[Proof of Theorem~\ref{thm:adaptive-moment-locality}: the graph bound]
			We first prove Equation~\eqref{eq:adaptive-graph-locality}.

			If $T=0$, $J_*=0$, or $d(X,\partial_E R)=\infty$, the two evolutions agree on $A$: in the latter two cases the dynamics factor over sites or over the connected components meeting $X$, respectively. Otherwise fix $0<t\le T$, $p>2$, and $\lambda\ge2$. Constants $C,c>0$ may change from line to line, with the dependencies stated in the theorem; the cutoff comparison below has $C,v_{\mathrm{LR}}$ depending only on the graph degree.

			By decomposing $A$ into its real and imaginary parts, it is enough, up to a universal factor absorbed into $C$, to prove the estimate for self-adjoint $A$. Trace duality then reduces the operator norm to testing against positive trace-class operators $\rho$ with $\tr\rho\le1$.
			
			We first derive a Duhamel identity comparing the full and truncated evolutions. Since their generators are unbounded, we begin with $\rho\in\cT_f$ in a fixed finite volume. Write $T_s=e^{s\cL}$ and $\tau_s=T_s^*$. Let $M:[0,t]\to\{2,3,\ldots\}$ be a piecewise constant function with value $M_j$ on each interval of a partition $0=t_0<\cdots<t_n=t$. The bounded-perturbation construction before Proposition~\ref{prop:bounded-interaction-lrb} shows that the closure of each $\cL_{R,M_j}$ generates a QMS. Define its predual propagator by
			\begin{equation*}
				T^{R,M}_{t,0}=e^{(t-t_{n-1})\cL_{R,M_n}}\cdots e^{(t_1-t_0)\cL_{R,M_1}},
			\end{equation*}
			and define $T^{R,M}_{u,s}$ analogously. We have
			\begin{equation*}
				[N_{\mathrm{tot}},H_{\mathrm{on}}]=[N_{\mathrm{tot}},H_e]=[N_{\mathrm{tot}},P_e^{(M)}]=0,
				\qquad [N_{\mathrm{tot}},a_x^\ell]=-\ell a_x^\ell,
			\end{equation*}
			where $H_{\mathrm{on}}$ is the on-site part of Equation~\eqref{eq:bose-hubbard}. Since $\rho\in\cT_f$, we have $\rho=Q_K\rho Q_K$ for some $K<\infty$, where $Q_K=\sum_{\|\mathbf n\|_1\le K}\ketbra{\mathbf n}{\mathbf n}$. The relations above show that both evolutions preserve the finite-dimensional space $Q_K\cH$. Their restricted generators are bounded, so the intervalwise product rule for $T^{R,M}_{t,s}T_s(\rho)$ gives
			\begin{equation}\label{eq:adaptive-local-duhamel}
			\begin{aligned}
				&\tr\!\left[\left(\tau_t(A)-(T^{R,M}_{t,0})^*(A)\right)\rho\right]=\int_0^t\tr\!\left[
				(T^{R,M}_{t,s})^*(A)
				(\cL-\cL_{R,M(s)})T_s(\rho)\right]ds.
			\end{aligned}
			\end{equation}
			For $q\in(2,2(\ell-1))$, the bounds in Equations~\eqref{eq:scheduled-hopping-tail} and \eqref{eq:proof-boundary-term} are integrable. For each $s>0$, approximation by bounded cutoffs, together with the tail estimates and uniform higher-moment bounds, gives trace-norm convergence of the unbounded commutators. Dominated convergence and $Q_K\rho Q_K\to\rho$ in trace norm extend the identity to every positive $\rho$ with $\tr\rho\le1$. Note that terms supported  in $V\setminus R$ vanish in Equation~\eqref{eq:adaptive-local-duhamel}, because $(T^{R,M}_{t,s})^*(A)$ is supported in $R$ and the outside dynamics is trace preserving. The remaining difference of generators has two parts. 
			
			\textbf{First, for all edges in $E_R$,} decompose the full hopping into the truncation tail $H_e-H_e^{(M(s))}$ and the bounded part $H_e^{(M(s))}$. For $e\in\partial_E R$, the bounded part is absent from $\cL_{R,M(s)}$, and its contribution is estimated in the second step. Proposition~\ref{main-prop:sobolev-regularizing} gives, for every $q>2$,
			\begin{equation}\label{eq:adaptive-local-moments}
				\tr[(N_z+\1)^qT_s(\rho)]\le b_q(s)^q,
			\end{equation}
			where $z\in V$, $\Delta_{\ell,p}$ is as in Proposition~\ref{main-prop:sobolev-regularizing}, and, for $p>2$ and $s>0$,
			\begin{equation}\label{eq:aux-bp-fct}
				b_p(s)=(Z_\kappa^*)^{1/p}\left[
					\left(\frac{2p}{(\ell-1)\gamma_{\min}\ell s}\right)^{1/(\ell-1)}
					+\Delta_{\ell,p}\right]\qquad\text{and}\qquad I_p(t)=\int_0^t b_p(s)\,ds.
			\end{equation}
			The same bound holds also for the dynamics restricted to $S\subseteq V$. Then, Lemma~\ref{lem:hopping-truncation-error} implies
			\begin{equation}\label{eq:scheduled-hopping-tail}
				\sum_{e\in E_R}\left|\tr\!\left[
				(T^{R,M}_{t,s})^*(A)
				[H_e-H_e^{(M(s))},T_s(\rho)]\right]\right|\le C J_*|E_R|\,\|A\|_\infty
				M(s)^{-(q-2)/2}b_q(s)^{q/2}.
			\end{equation}
			Here $C$ is a constant independent of $q$ (see Lemma~\ref{lem:hopping-truncation-error}). For an integer $\widetilde M\ge2$, choose
			\begin{equation*}
				M(s)=\min\{\widetilde M,\lceil\lambda b_p(s)\rceil\},\qquad s>0,
				\qquad M(0)=\widetilde M.
			\end{equation*}
			In the case $M(s)<\widetilde M$, and $p=q$ in Equation~\eqref{eq:scheduled-hopping-tail}, one has
			\begin{equation*}
				M(s)^{-(p-2)/2}b_p(s)^{p/2}
				\le\lambda^{-(p-2)/2}b_p(s).
			\end{equation*}
			Otherwise, for $M(s)=\widetilde M$, fix an auxiliary $p_0\in(2,2(\ell-1))$ and use the same estimate with $q=p_0$. Integrating gives the total hopping-tail contribution
			\begin{equation}\label{eq:adaptive-capped-tail}
				C J_*|E_R|\,\|A\|_\infty\left[
				\lambda^{-(p-2)/2}I_p(t)
				+\widetilde M^{-(p_0-2)/2}\int_0^t b_{p_0}(s)^{p_0/2}\,ds\right].
			\end{equation}
			Both integrals are finite. Indeed, $b_p(s)=\mathcal O_p(s^{-1/(\ell-1)}+1)$, $\ell>2$, and $p_0<2(\ell-1)$. Note that the second term tends to zero as $\widetilde M\to\infty$.

			\textbf{Second, for all edges in $\partial_E R$,} the bounded terms $H_e^{(M(s))}$ introduced above are absent from $\cL_{R,M(s)}$ and are estimated by Proposition~\ref{prop:bounded-interaction-lrb}. Its proof applies intervalwise to the piecewise constant schedule after the time change $r(s)=\int_0^s M(u)\,du$. In the new time variable the edge superoperators have norm at most $4J_*$, which gives the bound with integrated strength $\int_s^t M(u)\,du$.

			For $d(X,e)>0$, apply Proposition~\ref{prop:bounded-interaction-lrb} with $K_{e,s}=i[H_e^{(M(s))},\cdot]$ to obtain
			\begin{align*}
				\|K_{e,s}(T^{R,M}_{t,s})^*(A)\|_\infty
				&\le\frac12\|K_{e,s}\|_{\infty\to\infty}\|A\|_\infty
				\min\{|X|,|e|\}e^{-d(X,e)}
				\left(e^{v_{\mathrm{LR}}J_*\int_s^t M(u)\,du}-1\right)\\
				&\le C J_*M(s)\|A\|_\infty e^{-d(X,e)}
				\left(e^{v_{\mathrm{LR}}J_*\int_s^t M(u)\,du}-1\right),
			\end{align*}
			where $\|K_{e,s}\|_{\infty\to\infty}\le4J_*M(s)$ and $\min\{|X|,|e|\}\le2$.
			For edges intersecting $X$, use the elementary commutator bound $4J_*M(s)\|A\|_\infty$. The time integrals of this bound and the Lieb--Robinson estimate above are controlled, for $J_*>0$, by
			\begin{equation*}
				J_*\int_0^t M(s)e^{v_{\mathrm{LR}}J_*\int_s^t M(u)\,du}\,ds
				=\frac{e^{v_{\mathrm{LR}}J_*\int_0^t M(u)\,du}-1}
					{v_{\mathrm{LR}}}.
			\end{equation*}
			Summing the boundary edges therefore gives
			\begin{equation}\label{eq:proof-boundary-term}
				\int_0^t\!\!\sum_{e\in\partial_E R}\!\left|\tr\!\left[
				(T^{R,M}_{t,s})^*(A)
				[H_e^{(M(s))}\!,T_s(\rho)]\right]\right|ds\le C|\partial_E R|\,\|A\|_\infty e^{-d(X,\partial_E R)}
				\!\left(\!e^{v_{\mathrm{LR}}J_*\int_0^t M(u)\,du}\!-\!1\!\right)\!.
			\end{equation}
			Since $b_p(s)\ge1$ and $M(s)\le\lceil\lambda b_p(s)\rceil\le(\lambda+1)b_p(s)$, we have $\int_0^t M(s)\,ds\le(\lambda+1)I_p(t)$.

			Combining Equations~\eqref{eq:adaptive-local-duhamel}, \eqref{eq:adaptive-capped-tail}, and \eqref{eq:proof-boundary-term} bounds the difference between $\tau_t$ and $(T^{R,M}_{t,0})^*$. Apply the same Duhamel argument inside $R$ to compare $\tau_t^R$ with this propagator. There is no boundary contribution, and the tail is bounded by Equation~\eqref{eq:adaptive-capped-tail} with $E_R$ replaced by $E(R)=\{e\in E:e\subset R\}$. Since $|E(R)|\le|E_R|$, the triangle inequality and $\widetilde M\to\infty$ give, after adjusting $C$,
			\begin{equation}\label{eq:adaptive-physical-buffer-bound}
			\begin{aligned}
				\|\tau_t(A)-\tau_t^R(A)\|_\infty
				&\le C J_*|E_R|\,\|A\|_\infty
					\lambda^{-(p-2)/2}I_p(t)\\
				&\quad+C|\partial_E R|\,\|A\|_\infty e^{-d(X,\partial_E R)}
					\left(e^{v_{\mathrm{LR}}J_*(\lambda+1)I_p(t)}-1\right).
			\end{aligned}
			\end{equation}
			Only the error bound is passed to the limit. Write $d=d(X,\partial_E R)$. For $d\ge8v_{\mathrm{LR}}J_*I_p(t)$, take $\lambda=d/(4v_{\mathrm{LR}}J_*I_p(t))\ge2$. Then $v_{\mathrm{LR}}J_*(\lambda+1)I_p(t)\le3d/8$, which gives
			\begin{equation}\label{eq:adaptive-distance-bound}
			\begin{aligned}
				\|\tau_t(A)-\tau_t^R(A)\|_\infty
				&\le C J_*|E_R|\,\|A\|_\infty I_p(t)
					\left(\frac{4v_{\mathrm{LR}}J_*I_p(t)}{d}\right)^{(p-2)/2}\\
				&\quad+C|\partial_E R|\,\|A\|_\infty e^{-d/2}.
			\end{aligned}
			\end{equation}

			The explicit moment constants give $I_p(t)\le I_p(T)\le C(1+T)p2^p$ for $p>2$. For $q>0$, choose $p=2(q+1)$ in Equation~\eqref{eq:adaptive-distance-bound}. For $d\ge\max\{1,8v_{\mathrm{LR}}J_*I_p(T)\}$, use $|\partial_E R|\le|E_R|$ and
			\begin{equation*}
				\log[(Cp2^p)^{q+1}]\le C(1+q^2),\qquad
				\sup_{d\ge0}(1+d)^qe^{-d/2}\le Ce^{Cq^2}.
			\end{equation*}
			Below the threshold, $(1+d)/(1+T)\le Cp2^p$, so contractivity gives the same bound. Including $t=0$, we obtain, uniformly for every $q>0$,
			\begin{equation*}
				\sup_{0\le t\le T}\|\tau_t(A)-\tau_t^R(A)\|_\infty
				\le Ce^{Cq^2}|E_R|(1+T)^{q+1}(1+d)^{-q}\|A\|_\infty.
			\end{equation*}
			For $d>1+T$, choose $q=\log(d/(1+T))/(2C)$ to obtain
			\begin{equation*}
				C|E_R|(1+T)\|A\|_\infty
				\exp\!\left[-\frac{\log^2(d/(1+T))}{4C}\right].
			\end{equation*}
			Since $\log(2+u)\le\log3+\log u$ for $u\ge1$, this proves Equation~\eqref{eq:adaptive-graph-locality} after adjusting $C,c$. For $d\le1+T$, the same estimate follows from contractivity and $|E_R|\ge1$.
		\end{proof}
		\begin{figure}[htbp]
			\centering
            \vspace*{-2ex}
			\includegraphics[width=0.9\linewidth]{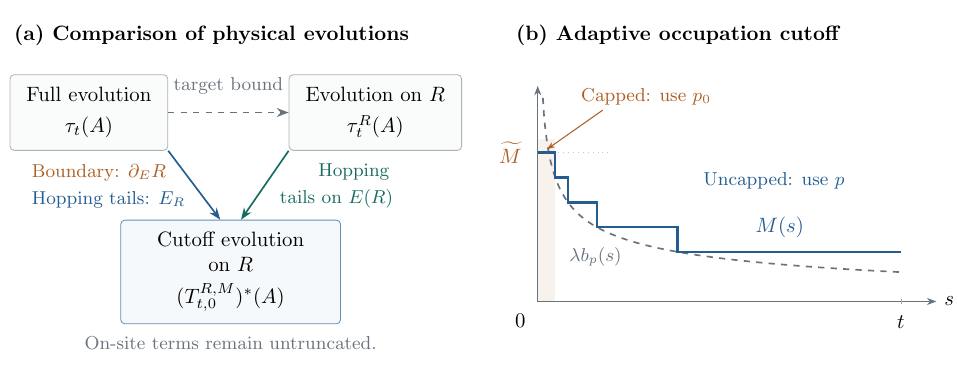}
            \vspace*{-2ex}
			\caption{The adaptive locality proof. (a) Comparing both physical evolutions with the same cutoff evolution on $R$ gives hopping tails controlled by moments and boundary terms controlled by a Lieb--Robinson bound. (b) Schematic cutoff $M(s)=\min\{\widetilde M,\lceil\lambda b_p(s)\rceil\}$. The capped interval uses $2<p_0<2(\ell-1)$, and the remainder uses $p>2$. For $\ell>2$, the integrated cutoff is bounded uniformly in $\widetilde M$.}
			\label{fig:adaptive-locality-proof}
		\end{figure}
		\begin{proof}[Proof of Theorem~\ref{thm:adaptive-moment-locality}: the lattice bound]
		We now prove Equation~\eqref{eq:adaptive-log-gaussian-locality}. Put $R=X[r]$. If $T=0$, $J_*=0$, or the ball has no outgoing edges, the two evolutions agree on $A$. Otherwise $d(X,\partial_E R)=r$ and
		\begin{equation*}
			|E_R|,|\partial_E R|\le2D|R|\le2D|X|(2r+1)^D.
		\end{equation*}
		The explicit moment constants give $I_p(t)\le I_p(T)\le C(1+T)p2^p$ for $p\ge3$, with $C$ depending only on the fixed model bounds. For $q>0$, choose $p=2(D+q+1)$ in Equation~\eqref{eq:adaptive-distance-bound}. Above its radius threshold, the truncation term is bounded by
		\begin{equation*}
			C|X|(Cp2^p)^{p/2}(1+T)^{D+q+1}r^{-q}\|A\|_\infty.
		\end{equation*}
		Since $\log[(Cp2^p)^{p/2}]\le C(1+q^2)$ and $\sup_{r\ge0}(1+r)^{D+q}e^{-r/2}\le Ce^{Cq^2}$, both terms are bounded by $Ce^{Cq^2}|X|(1+T)^{D+q+1}(1+r)^{-q}\|A\|_\infty$. The radius threshold is at most $C(1+T)p2^p$, so contractivity covers smaller radii with the same bound. Replace $q$ by $(D+1)q$ and apply
		\begin{equation*}
			\min\{2,u\}\le2^{D/(D+1)}u^{1/(D+1)},\qquad u\ge0,
		\end{equation*}
		to the normalized error. Contractivity and $|X|^{1/(D+1)}\le|X|$ then give, after enlarging $C$,
		\begin{equation*}
			\sup_{0\le t\le T}\|\tau_t(A)-\tau_t^R(A)\|_\infty
			\le Ce^{Cq^2}|X|(1+T)^{q+1}(1+r)^{-q}\|A\|_\infty.
		\end{equation*}

		For $r>1+T$, choose $q=\log(r/(1+T))/(2C)$ in this estimate. Using $(1+r)^{-q}\le r^{-q}$ gives
		\begin{equation*}
			C|X|(1+T)\|A\|_\infty
			\exp\!\left[-\frac{\log^2(r/(1+T))}{4C}\right].
		\end{equation*}
		Since $\log(2+u)\le\log3+\log u$ for $u\ge1$, this proves Equation~\eqref{eq:adaptive-log-gaussian-locality} after adjusting $C,c$. For $r\le1+T$, the same estimate follows from contractivity.
	\end{proof}

	\subsection{Proof of the four-photon bound}
\begin{proof}[Proof of Proposition~\ref{prop:cat-code-polynomial-lrb}]
		Following the calculation in \cite[Prop.~B.6, Lem.~B.5]{Rico.2025} applied to $\cL_R^\gamma$ gives, for every $k\ge1$,
		\begin{equation}\label{eq:cat-core-drift}
			\tr[M_{v,k}\gamma^{-1}\cL_R^\gamma(\rho)]
				\le-\mu\tr[M_{v,k+3}\rho]+c_k\tr[\rho],
				\qquad \rho\in\cT_f,\quad\rho\ge0,
		\end{equation}
		with $\mu=2\gamma_-$ and $c_k<\infty$ independent of $R$, $v$, and $\gamma$. This concentrated Sobolev regularization condition implies by \cite[Thm.~3.12]{Rico.2025} the generation theory as well as the Sobolev regularization property for the generated $C_0$-semigroup $e^{s\cL_R^\gamma}$:
		\begin{equation}\label{eq:cat-local-moments}
			\tr[(N_x+\1)^k e^{s\cL_R^\gamma}(\rho)]^{1/k}
				\le C_k^{\mathrm{mom}}\bigl((\gamma s)^{-1/3}+1\bigr)
				\le b_k(s)\coloneqq C_k^{\mathrm{mom}}(s^{-1/3}+1)
		\end{equation}
		for every state $\rho$, $s>0$, $x\in R$, and $C_k^{\mathrm{mom}}\ge1$ only depending on $k$ and the fixed model bounds. In the cited commutator estimate, occupation shifts of size at most four give coefficients bounded by $Ck5^{2k}$. The coefficient before scalar absorption is therefore at most $Ce^{Ck}$, so $(c_k/\mu)^{1/(k+3)}\le Ce^{Ck}$ and $C_k^{\mathrm{mom}}\le Ce^{Ck}$, with $C$ independent of $k$.

		We next verify the polynomial truncation estimate. For $M\ge8$, put
		\begin{equation*}
			S_e=\sum_{z\in e}(N_z+\1),\qquad
			Q_e^{(M)}=\1-P_e^{(M)},\qquad
			\widetilde Q_e^{(M)}=\1-P_e^{(M-4)}.
		\end{equation*}
		For a monomial $B$ of degree at most four, let $d_z$ count its creation and annihilation factors at $z$. The Fock-basis amplitudes give, for $\psi=\sum_{\mathbf n}\psi_{\mathbf n}\ket{\mathbf n}\in\cH_f$,
		\begin{equation*}
			\|B\psi\|^2\le C\sum_{\mathbf n}\prod_{z\in e}(n_z+1)^{d_z}|\psi_{\mathbf n}|^2
			\le C\|S_e^2\psi\|^2,\qquad \sum_{z\in e}d_z\le4.
		\end{equation*}
		Thus $BS_e^{-2}$ extends to a bounded operator. For comparison, see \cite[Lem.~E.2]{Gondolf.2024} for a related single-mode polynomial weight estimate. Moreover,
		\begin{equation*}
			B-P_e^{(M)}BP_e^{(M)}
			=BQ_e^{(M)}+Q_e^{(M)}BP_e^{(M)},
			\qquad
			Q_e^{(M)}BP_e^{(M)}
			=Q_e^{(M)}BP_e^{(M)}\widetilde Q_e^{(M)}.
		\end{equation*}
		The second identity follows because $B$ changes each occupation by at most four. Since the projections commute with $S_e$ and $Q_e^{(M)}\le\widetilde Q_e^{(M)}$, for every $\psi\in\cH_f$,
		\begin{equation*}
			\|(B-P_e^{(M)}BP_e^{(M)})\psi\|
			\le\|BQ_e^{(M)}\psi\|
				+\|Q_e^{(M)}BP_e^{(M)}\widetilde Q_e^{(M)}\psi\|
			\le C\|S_e^2\widetilde Q_e^{(M)}\psi\|.
		\end{equation*}
		Squaring gives the quadratic-form inequality
		\begin{equation*}
			(B-P_e^{(M)}BP_e^{(M)})^\dagger
			(B-P_e^{(M)}BP_e^{(M)})\le C\widetilde Q_e^{(M)}S_e^4\widetilde Q_e^{(M)}.
		\end{equation*}
		For $M\ge8$, the range of $\widetilde Q_e^{(M)}$ satisfies $S_e\ge M/2$. Hence, for $p>4$,
		\begin{equation*}
			\widetilde Q_e^{(M)}S_e^4\widetilde Q_e^{(M)}
			\le C_pM^{-(p-4)}\sum_{z\in e}(N_z+\1)^p.
		\end{equation*}
		Then, Cauchy--Schwarz shows, for every $\sigma\geq0$ with $\tr\sigma\le1$ and finite local $p$-moments,
		\begin{equation}\label{eq:cat-polynomial-tail}
			\|(H_e-H_e^{(M)})\sigma\|_1+\|\sigma(H_e-H_e^{(M)})\|_1\le C_pJ_{\mathrm{cat}}M^{-(p-4)/2}\sum_{z\in e}\tr[(N_z+\1)^p\sigma]^{1/2}.
		\end{equation}
		Enlarging $C_p$ covers $2\le M<8$. The scalar weight estimates above give $C_p\le Ce^{Cp}$. The same weighted bound gives $\|H_e^{(M)}\|_\infty\le CJ_{\mathrm{cat}}M^2$ and therefore $\|[-iH_e^{(M)},\cdot]\|_{\infty\to\infty}\le CJ_{\mathrm{cat}}M^2$.

		We now follow the two-part comparison in the proof of Theorem~\ref{thm:adaptive-moment-locality}. Let $T_s^\gamma=e^{s\cL^\gamma}$, and let $\cL_{R,M}^\gamma$ be obtained from $\cL_R^\gamma$ by replacing every $H_e$ with $H_e^{(M)}$. For a piecewise constant cutoff schedule $M$, denote the corresponding predual propagator by $T^{\gamma,R,M}_{t,s}$. Equation~\eqref{eq:adaptive-local-duhamel} then holds with $T_s$, $T^{R,M}_{t,s}$, $\cL$, and $\cL_{R,M(s)}$ replaced by $T_s^\gamma$, $T^{\gamma,R,M}_{t,s}$, $\cL^\gamma$, and $\cL_{R,M(s)}^\gamma$. For finite-Fock inputs, all-order Sobolev stability and \cite[Lem.~D.2]{Rico.2025} provide a sufficiently high invariant Sobolev space contained in the relevant generator domains, which justifies the intervalwise product rule without particle-number invariance.

		\textbf{First, for all edges in $E_R$,} decompose $H_e$ into the truncation tail $H_e-H_e^{(M(s))}$ and the bounded part $H_e^{(M(s))}$. For boundary edges, the bounded part is absent from the cutoff buffer generator and is estimated below. For an integer $\widetilde M\ge2$, choose
		\begin{equation*}
			M(s)=\min\{\widetilde M,\lceil\lambda b_p(s)\rceil\},\quad s>0,
			\qquad M(0)=\widetilde M.
		\end{equation*}
		Use Equation~\eqref{eq:cat-polynomial-tail} at order $p$ where $M(s)<\widetilde M$ and at order $5$ where $M(s)=\widetilde M$. By Equation~\eqref{eq:cat-local-moments}, the integrated tail is at most
		\begin{equation*}
			C_pJ_{\mathrm{cat}}|E_R|\,\|A\|_\infty
			\left[\lambda^{-(p-4)/2}\int_0^t b_p(s)^2\,ds
				+\widetilde M^{-1/2}\int_0^t b_5(s)^{5/2}\,ds\right].
		\end{equation*}
		Both integrals are finite because $b_p(s)^2=\mathcal O_p(s^{-2/3}+1)$ and $b_5(s)^{5/2}=\mathcal O(s^{-5/6}+1)$. The second term tends to zero as $\widetilde M\to\infty$. For each $s>0$, approximation by bounded cutoffs, together with the tail estimates and uniform higher-moment bounds, gives trace-norm convergence of the unbounded commutators. These bounds also extend the Duhamel identity to arbitrary states by positive finite-Fock approximation.

		\textbf{Second, for all edges in $\partial_E R$,} the bounded terms $H_e^{(M(s))}$ introduced above are absent from the cutoff buffer generator. Apply Proposition~\ref{prop:bounded-interaction-lrb} first with $K_{e,s}=i[H_e^{(M(s))},\cdot]$. Its proof applies intervalwise after the time change $r(s)=\int_0^sM(u)^2\,du$. In the new time variable the edge superoperators have norm at most $CJ_{\mathrm{cat}}$, because
		\begin{equation*}
			\|K_{e,s}\|_{\infty\to\infty}\le CJ_{\mathrm{cat}}M(s)^2.
		\end{equation*}
		Summing and integrating the resulting Lieb--Robinson estimate as in Equation~\eqref{eq:proof-boundary-term} gives
		\begin{equation*}
			C|\partial_E R|\,\|A\|_\infty e^{-d(X,\partial_E R)}
			\left(e^{C_pJ_{\mathrm{cat}}\int_0^tM(s)^2\,ds}-1\right).
		\end{equation*}
		The constants are independent of the on-site rates. Since
		\begin{equation*}
			\int_0^t M(s)^2\,ds
				\le(\lambda+1)^2\int_0^t b_p(s)^2\,ds
				\le6(C_p^{\mathrm{mom}})^2(\lambda+1)^2(t^{1/3}+t),
		\end{equation*}
		this is the boundary term in Equation~\eqref{eq:cat-code-lrb}, after adjusting $C_p$. The comparison above bounds the difference between $\tau_t^\gamma$ and $(T^{\gamma,R,M}_{t,0})^*$. Apply the same Duhamel argument inside $R$ to compare $\tau_t^{\gamma,R}$ with this propagator. This comparison has no boundary term, and its truncation tail is bounded by the first estimate with $E_R$ replaced by $E(R)$. The triangle inequality, $|E(R)|\le|E_R|$, and $\widetilde M\to\infty$ prove Equation~\eqref{eq:cat-code-lrb} by trace duality. Only the truncation constant and $(C_p^{\mathrm{mom}})^2$ enter its $p$ dependence, so $C_p\le Ce^{Cp}$.

		For the lattice assertion, if $J_{\mathrm{cat}}=0$ or the ball has no outgoing edges, the two evolutions agree on $A$. Otherwise, $d(X,\partial_E R)=r$ and
		\begin{equation*}
			|E_R|,|\partial_E R|\le2D|R|\le2D|X|(2r+1)^D.
		\end{equation*}
		For sufficiently large $r$, choose
		\begin{equation*}
			\lambda=\left(\frac{r}{8C_pJ_{\mathrm{cat}}(T^{1/3}+T)}\right)^{1/2}\ge2.
		\end{equation*}
		Since $(\lambda+1)^2\le9\lambda^2/4$, the boundary term is bounded by a constant times $(1+r)^De^{-r/2}\|A\|_\infty$. The remaining term is $\mathcal O_{p,T,X}(r^{D-(p-4)/4})\|A\|_\infty$. Taking $p>4(D+q)+4$ proves Equation~\eqref{eq:cat-locality-lattice} for these radii. For the remaining radii, contractivity gives $\|\tau_t^\gamma(A)-\tau_t^{\gamma,R}(A)\|_\infty\le2\|A\|_\infty$, which is absorbed by enlarging $C_{q,T,X}$.
	\end{proof}

	\begin{rmk}[Higher multiphoton loss]\label{rmk:standard-cat-code-scope}
		Given all-order moment bounds with root singularity $s^{-1/(\ell-1)}$, the same argument for spatial polynomials of degree $d$ requires $d<2(\ell-1)$, since the integrated interaction strength involves $s^{-d/(2(\ell-1))}$. For $\ell=2$, quadratic spatial interactions are critical. The present state-independent argument does not cover this case.
	\end{rmk}

\addtocontents{toc}{\protect\setcounter{tocdepth}{3}}
\section{Consequences}\label{sec:appl}
	We write $\|\Phi\|_\diamond$ for the completely bounded trace norm and measure channel errors by $\|\Phi-\Psi\|_\diamond$.

\subsection{Local channels and the thermodynamic limit}\label{subsec:local-projected-simulation}
	Denote the dissipative Bose--Hubbard evolution on $V$ by $T_t=e^{t\cL_V}$ and its restriction to $R$, extended by the identity outside $R$, by $T_t^R$.

	\begin{cor}[Local channel approximation]\label{cor:local-channel-approximation}
		Under the hypotheses of Theorem~\ref{thm:adaptive-moment-locality} on nearest-neighbor subgraphs of $\mathbb Z^D$, fix a finite $X\subset V$ and $T>0$, and put $R=X[r]$ for an integer $r\ge0$. There are constants $C,c>0$ such that
		\begin{equation}\label{eq:local-channel-approximation}
			\sup_{0\le t\le T}
			\left\|\tr_{V\setminus X}\circ T_t
				-\tr_{V\setminus X}\circ T_t^R\right\|_\diamond
			\le C|X|(1+T)
			\exp\!\left[-c\log^2\!\left(2+\frac{r}{1+T}\right)\right].
		\end{equation}
		The constants depend only on $D$, $\ell$, $J_*$, $\Delta_G$, the loss-rate bounds, and the fixed Sobolev weights. They are independent of $T$, $X$, $r$, the volume, input state, and reference dimension, and require no bound on the on-site Bose--Hubbard coefficients. For two such generators agreeing on $R$, the difference of their output channels on $X$ satisfies twice this bound, with common uniform coefficient bounds.
	\end{cor}
	\begin{proof}
		Complete trace-norm duality gives
		\begin{equation*}
			\left\|\tr_{V\setminus X}\circ T_t
				-\tr_{V\setminus X}\circ T_t^R\right\|_\diamond=\sup_{\cH_{\mathrm{ref}}}\sup_{\rho\in\cS(\cH_V\otimes\cH_{\mathrm{ref}})}\sup_{\substack{O\in\cB(\cH_X\otimes\cH_{\mathrm{ref}})\\\|O\|_\infty\le1}}
				\left|\tr\!\left[((\tau_t-\tau_t^R)\otimes\id_{\cH_{\mathrm{ref}}})(O)\rho\right]\right|.
		\end{equation*}
		The first supremum runs over finite-dimensional reference spaces $\cH_{\mathrm{ref}}$. Since the generator acts trivially on $\cH_{\mathrm{ref}}$, adjoining it to a site in $X$ leaves all moments, supports, and interaction norms unchanged. Thus Theorem~\ref{thm:adaptive-moment-locality} applies uniformly in $\cH_{\mathrm{ref}}$, and the duality identity gives Equation~\eqref{eq:local-channel-approximation}. If two generators agree on $R$, their restricted evolutions coincide, and the triangle inequality through this common evolution gives twice the bound.
	\end{proof}

	\begin{cor}[Local stability under Hamiltonian perturbations]\label{cor:local-hamiltonian-stability}
		Let $T_t$ and $\widetilde T_t$ satisfy the hypotheses of Corollary~\ref{cor:local-channel-approximation} on the same graph, with identical loss terms and a common hopping bound $J_*$. Suppose that $\ell>3$ and
		\begin{equation*}
			|J_{xy}-\widetilde J_{xy}|\le\varepsilon,\qquad
			|\mu_x-\widetilde\mu_x|+|U_x-\widetilde U_x|\le\varepsilon
		\end{equation*}
		for every edge $(x,y)\in E$ and site $x\in V$. Then, for every finite $X\subset V$ and $T>0$,
		\begin{equation}\label{eq:local-hamiltonian-stability}
			\sup_{0\le t\le T}\left\|\tr_{V\setminus X}\circ(T_t-\widetilde T_t)\right\|_\diamond
			\le C\varepsilon |X|^2(1+T)^{D+1}
			\left(T^{1-2/(\ell-1)}+T\right).
		\end{equation}
		The constant has the dependencies stated in Corollary~\ref{cor:local-channel-approximation}. If $U_x=\widetilde U_x$ for every $x$, the assertion holds for $\ell>2$ with $T^{1-2/(\ell-1)}$ replaced by $T^{1-1/(\ell-1)}$.
	\end{cor}
	\begin{proof}
		Fix $A\in\cB(\cH_X)$ and write $\widetilde H-H=\sum_Z h_Z$, where $Z$ is a site or an edge. Proposition~\ref{main-prop:sobolev-regularizing} at order $k=4$ and Cauchy--Schwarz ineqaulity give
		\begin{equation*}
			\|[h_Z,\widetilde T_s(\rho)]\|_1
			\le C\varepsilon\left(1+s^{-2/(\ell-1)}\right)
		\end{equation*}
		for every state $\rho$ and $s>0$. Duhamel's formula, justified on finite total-number sectors as in Equation~\eqref{eq:adaptive-local-duhamel}, pairs this commutator with $\tau_{t-s}(A)$. For $Z$ disjoint from $X$, replace this observable by its restricted evolution on $X[r]$ with $r=\mathrm{dist}(X,Z)-1$. The replacement commutes with $h_Z$, and Corollary~\ref{cor:local-channel-approximation} bounds the remaining term. For $Z$ meeting $X$, use contractivity. Summing the spatial decay using $|X[r]|\le |X|(2r+1)^D$ gives
		\begin{equation*}
			\left|\tr[A(T_t-\widetilde T_t)(\rho)]\right|
			\le C\varepsilon |X|^2\|A\|_\infty
			\int_0^t\left(1+s^{-2/(\ell-1)}\right)(1+t-s)^{D+1}\,ds.
		\end{equation*}
		The integral is finite for $\ell>3$, extends the identity to arbitrary inputs, and yields Equation~\eqref{eq:local-hamiltonian-stability}. The argument is uniform in a finite reference as in Corollary~\ref{cor:local-channel-approximation}. If the  coefficients agree, each $h_Z$ has degree at most two, so the moment singularity is $s^{-1/(\ell-1)}$ instead.
	\end{proof}

	For an infinite graph $G$, the quasi-local algebra is the operator-norm closure of the union of $\cB(\cH_Y)$ over all finite $Y\subset G$, with local observables extended by the identity outside their support.

	\begin{cor}[Thermodynamic limit]\label{cor:thermodynamic-limit}
		Let $G$ be an infinite nearest-neighbor subgraph of $\mathbb Z^D$, and let the finite-volume generators be consistent restrictions of a fixed interaction on $G$, with the bounds in Theorem~\ref{thm:adaptive-moment-locality} uniform in the volume. Then there is a unique contractive, completely positive, unital semigroup $(\tau_t)_{t\ge0}$ on the quasi-local algebra. For every $A\in\cB(\cH_X)$ with finite $X$, $T>0$, integer $r\ge0$, and finite $V\supseteq X[r]$,
		\begin{equation*}
			\sup_{0\le t\le T}\|\tau_t^V(A)-\tau_t(A)\|_\infty
			\le 2C|X|(1+T)\|A\|_\infty
			\exp\!\left[-c\log^2\!\left(2+\frac{r}{1+T}\right)\right].
		\end{equation*}
		Here $C,c$ have the dependencies stated in Theorem~\ref{thm:adaptive-moment-locality}. In particular, the finite-volume evolutions converge uniformly on $[0,T]$. No norm continuity at $t=0$ is asserted for all bounded local observables.
	\end{cor}
	\begin{proof}
		For finite volumes $V,W$ containing $R=X[r]$, Theorem~\ref{thm:adaptive-moment-locality}  gives
		\begin{equation*}
			\sup_{0\le t\le T}\|\tau_t^V(A)-\tau_t^W(A)\|_\infty
			\le2C|X|(1+T)\|A\|_\infty
			\exp\!\left[-c\log^2\!\left(2+\frac{r}{1+T}\right)\right].
		\end{equation*}
		The right side tends to zero as $r\to\infty$. Passing $W\uparrow G$ gives the displayed rate. Contractivity extends the limit to the quasi-local algebra, while complete positivity, unitality, and the semigroup law pass to the limit.
	\end{proof}

\subsection{Correlation spreading and preparation time}\label{subsec:correlation-spreading}
	Local channel approximation also limits the creation of correlations between distant regions. This gives lower bounds on the time needed to prepare correlated states from product inputs, independently of their initial occupation.

	\begin{cor}[Correlation spreading from product inputs]\label{cor:correlation-spreading}
		Under the hypotheses of Corollary~\ref{cor:local-channel-approximation}, let $\rho=\bigotimes_{x\in V}\rho_x$ be a product state and write $\rho_{t,Z}=\tr_{V\setminus Z}T_t(\rho)$, with $XY$ denoting $X\cup Y$. For nonempty disjoint regions $X,Y\subset V$ and $t\ge0$,
		\begin{equation}\label{eq:correlation-spreading}
			\|\rho_{t,XY}-\rho_{t,X}\otimes\rho_{t,Y}\|_1
			\le C(|X|+|Y|)(1+t)
			\exp\!\left[-c\log^2\!\left(2+\frac{\mathrm{dist}(X,Y)}{1+t}\right)\right].
		\end{equation}
		The constants have the dependencies stated in Corollary~\ref{cor:local-channel-approximation}. No initial moment bound is required.
	\end{cor}
	\begin{proof}
		Put $d=\mathrm{dist}(X,Y)$. For $d\ge2$, choose $r=\lfloor(d-2)/2\rfloor$ and $R=X[r]\cup Y[r]$. No edge joins the two neighborhoods, so the restricted output is $\sigma_X\otimes\sigma_Y$, where $\sigma_Z=\tr_{V\setminus Z}T_t^R(\rho)$ for $Z=X,Y$. Corollary~\ref{cor:local-channel-approximation} bounds its distance from $\rho_{t,XY}$. Contractivity under partial trace and the triangle inequality give
		\begin{equation*}
			\|\rho_{t,XY}-\rho_{t,X}\otimes\rho_{t,Y}\|_1
			\le3\|\rho_{t,XY}-\sigma_X\otimes\sigma_Y\|_1.
		\end{equation*}
		This proves Equation~\eqref{eq:correlation-spreading} after adjusting $C,c$. The case $d=1$ follows from the trace-distance bound $2$.
	\end{proof}

	\emph{Preparation time.} For observables $A_X,B_Y$ of norm at most one, define $\operatorname{Cov}_{\sigma}(A_X,B_Y)=\tr[A_XB_Y\sigma]-\tr[A_X\sigma]\tr[B_Y\sigma]$. Suppose a target state $\omega$ on $XY$ satisfies
	\begin{equation*}
		|\operatorname{Cov}_{\omega}(A_X,B_Y)|\ge\chi>0,
		\qquad \|\rho_{t,XY}-\omega\|_1\le\varepsilon<\chi/6.
	\end{equation*}
	The covariance changes by at most $3\varepsilon$, so its magnitude in $\rho_{t,XY}$ is at least $\chi/2$ and is bounded by the left side of Equation~\eqref{eq:correlation-spreading}. For fixed $|X|,|Y|$ and $\chi$, this implies, for sufficiently large $d=\mathrm{dist}(X,Y)$,
	\begin{equation}\label{eq:correlation-preparation-time}
		t\ge c_0d\exp\!\left[-C_0\sqrt{\log(2+d)}\right].
	\end{equation}
	Indeed, when $1+t\le d$, Equation~\eqref{eq:correlation-spreading} gives $c\log^2(2+d/(1+t))\le\log(2C(|X|+|Y|)d/\chi)$, and the case $1+t>d$ is immediate. The constants depend only on the fixed support sizes, $\chi$, and the model bounds. In particular, for every $\eta>0$, $t\ge c_\eta d^{1/(1+\eta)}$. Thus arbitrarily large initial occupations cannot circumvent this almost-linear preparation-time bound within the considered dynamics.

\subsection{Local cat-code adiabatic limit}\label{subsec:local-adiabatic-limit}
	Consider the model of Proposition~\ref{prop:cat-code-polynomial-lrb} on $\mathbb Z^D$, with $\gamma_x^{\mathrm{cat}}=1$ and on-site Hamiltonian degree at most four. Thus, on a finite volume $V$ and for $P_x$ being the projection onto $\ker L_x$,
	\begin{equation*}
		\cL_V^\gamma=-i[H_V^{\mathrm{cat}},\cdot]
			+\gamma\sum_{x\in V}\cL[L_x].
	\end{equation*}
	For every finite region $R\subseteq V$, set $P_R=\bigotimes_{x\in R}P_x$ and $\mathcal K_R^{\mathrm{proj}}=-i[P_RH_R^{\mathrm{cat}}P_R,\cdot]$.

	\begin{thm}[Local cat-code adiabatic estimate]\label{thm:local-cat-adiabatic-estimate}
		There is $C\ge1$ such that, for every finite $X\subset V$, $T>0$, $0<\eta<1$, $\gamma\ge1$, and state $\rho=P_V\rho P_V$,
		\begin{equation}\label{eq:local-cat-adiabatic-bound}
			\sup_{0\le t\le T}\left\|\tr_{V\setminus X}
				\left(e^{t\cL_V^\gamma}(\rho)
					-e^{t\mathcal K_V^{\mathrm{proj}}}(\rho)\right)\right\|_1
				\le Ce^{C/\eta}|X|^2(1+T)^{2D+1}\gamma^{-1+\eta}.
		\end{equation}
		The constant depends only on $D$, the fixed Sobolev weight, and the uniform amplitude and Hamiltonian coefficient bounds. It is independent of $T$, $\eta$, $X$, the volume, and $\gamma$.
	\end{thm}
	\begin{proof}
		Put $R=X[r]$. The reduced state $\tr_{V\setminus R}\rho$ is supported on $P_R$, so Lemma~\ref{lem:cat-finite-adiabatic} bounds the difference between the dissipative and projected evolutions on $R$ by $C(T|R|^2+|R|)/\gamma$. Its code moment estimate gives $\|P_xP_yH_eP_xP_y\|_\infty\le C$ uniformly for $e=(x,y)$. On the tensor-product code space, Proposition~\ref{prop:bounded-interaction-lrb} and Duhamel's formula therefore bound the projected restriction error by $C|\partial_E R|e^{CT-r}$. Choose
		\begin{equation*}
			q=\frac{2D(1-\eta)}{\eta},\qquad p=4(D+q)+4,\qquad
			L_\eta=e^{C/\eta},\qquad
			r=\left\lceil L_\eta(1+T)\gamma^{\eta/(2D)}\right\rceil.
		\end{equation*}
		Since $C_p\le Ce^{Cp}$, we may enlarge $C$ so that $L_\eta\ge64C_pJ_{\mathrm{cat}}$, $L_\eta\ge2(D+q)$, and $CT-r\le-r/2$. For $J_{\mathrm{cat}}>0$, apply Equation~\eqref{eq:cat-code-lrb} with
		\begin{equation*}
			\lambda=\left(\frac{r}{8C_pJ_{\mathrm{cat}}(T^{1/3}+T)}\right)^{1/2}
			\ge2\gamma^{\eta/(4D)}.
		\end{equation*}
		Using $|E_R|,|\partial_E R|\le2D|R|\le2D|X|(2r+1)^D$, trace duality and the triangle inequality bound the left side of Equation~\eqref{eq:local-cat-adiabatic-bound} by
		\begin{equation*}
		\begin{aligned}
			&C C_pJ_{\mathrm{cat}}|X|r^D\lambda^{-(p-4)/2}(T^{1/3}+T)
				+C|X|r^De^{-r/2}+\frac{C(T|X|^2r^{2D}+|X|r^D)}{\gamma}.
		\end{aligned}
		\end{equation*}
		If $J_{\mathrm{cat}}=0$, the restriction errors vanish. Since $(p-4)/4-D=q$ and $q\eta/(2D)=1-\eta$, the first term is at most $C|X|[L_\eta(1+T)]^{D+1}\gamma^{-1+\eta}$. The second has the same bound, using $z^{D+q}e^{-L_\eta z/2}\le1$ for $z\ge1$. The last term is at most $C|X|^2L_\eta^{2D}(1+T)^{2D+1}\gamma^{-1+\eta}$. As $D+1\le2D$, these estimates prove the assertion after enlarging $C$.
	\end{proof}

\subsection{Global approximation and simulation}\label{subsec:global-simulation}
	We present a hybrid bosonic--qubit simulation of the dissipative Bose--Hubbard model under the hypotheses of Theorem~\ref{thm:adaptive-moment-locality}, with $\ell>3$. Physical multiphoton loss first regularizes an arbitrary bosonic input. After an occupation test and coherent encoding, a qubit circuit simulates the remaining evolution. This gives a diamond-norm error bound uniform over all input states. For $0<\delta<t$ and an integer cutoff $M\ge2$, the scheme proceeds as follows.\\[-1.5ex]
    \begin{figure}[H]
		\centering
        \vspace*{-2ex}
        \includegraphics[width=0.9\linewidth]{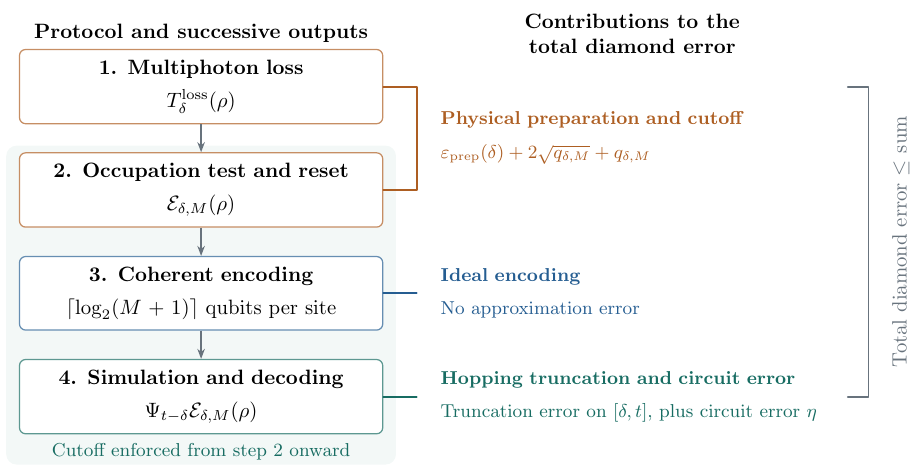}
        \vspace*{-2ex}
		\caption{Hybrid bosonic--qubit simulation at time $t$. The four steps implement $\Psi_{t-\delta}\mathcal E_{\delta,M}$. The right column summarizes the error contributions in Equation~\eqref{eq:global-bh-unrestricted-error}, whose sum bounds the diamond distance from $T_t$. Ideal encoding adds no error. Shading marks the stages after the occupation cutoff is enforced.}
		\label{fig:global-simulation-comparison}
	\end{figure}
    
	\begin{enumerate}[leftmargin=*,itemsep=1pt,topsep=3pt]
		\item Apply only the physical multiphoton loss for short time $\delta$ to regularize the input.
		\item At each site, test whether the occupation is at most $M$, without resolving occupations within either outcome. Keep the state if every test passes. Otherwise, reset the system to the vacuum.
		\item Coherently encode each occupation $0,\ldots,M$ in a binary register.
		\item Simulate the compressed Hamiltonian and loss for time $t-\delta$, then decode.
	\end{enumerate}~\\[-1.5ex]
	The same cutoff controls both the hopping and the register size. The on-site Hamiltonian is included in the circuit, and its omission during preparation is part of the error below.

	\begin{prop}[Dissipative preparation]\label{prop:global-onsite-preparation}
		Let $\delta>0$ and let $M\ge2$ be an integer. Write $T_u^{\mathrm{loss}}=e^{u\sum_{x\in V}\gamma_x\cL[a_x^\ell]}$, $P_V^{(M)}=\prod_{x\in V}P_x^{(M)}$, and $\cH_M=\operatorname{ran}P_V^{(M)}$. We identify states on $\cH_M$ with their natural embedding in $\cT_1(\cH_V)$. With $\ket{0_V}=\bigotimes_{x\in V}\ket0$, the preparation channel is
		\begin{equation}\label{eq:global-encoder}
			\mathcal E_{\delta,M}(\rho)
			=P_V^{(M)}T_\delta^{\mathrm{loss}}(\rho)P_V^{(M)}
			+\tr[(\1-P_V^{(M)})T_\delta^{\mathrm{loss}}(\rho)]\ketbra{0_V}{0_V}.
		\end{equation}
		For $u>0$ and $k\ge1$, set
		\begin{equation}\label{eq:global-preparation-quantities}
		\begin{aligned}
			B_k^{\mathrm{loss}}(u)
			&=\left(\frac{2k}{(\ell-1)\gamma_{\min}\ell u}\right)^{k/(\ell-1)}
			  +[\ell(\ell+1)]^k,\\
			q_{\delta,M}&=\min\left\{1,|V|\inf_{k\ge1}\frac{B_k^{\mathrm{loss}}(\delta)}{(M+2)^k}\right\},\\
			\varepsilon_{\mathrm{prep}}(\delta)
			&=\left(4J_*|E|+2\sum_{x\in V}|\mu_x|\right)
			  \int_0^\delta\sqrt{B_2^{\mathrm{loss}}(u)}\,du+\sum_{x\in V}|U_x|\int_0^\delta\sqrt{B_4^{\mathrm{loss}}(u)}\,du.
		\end{aligned}
		\end{equation}
		Then
		\begin{equation}\label{eq:global-preparation-error}
			\|T_\delta-\mathcal E_{\delta,M}\|_\diamond
			\le\varepsilon_{\mathrm{prep}}(\delta)+2\sqrt{q_{\delta,M}}+q_{\delta,M}.
		\end{equation}
	\end{prop}
	\begin{proof}
		Fix an arbitrary finite reference $\cH_{\mathrm{ref}}$ and a state $\rho\in\cT_f\otimes\cT_1(\cH_{\mathrm{ref}})$. Extend all system maps and operators by the identity on $\cH_{\mathrm{ref}}$.

		We first compare the full evolution with pure loss during preparation. Apply Proposition~\ref{main-prop:sobolev-regularizing} at each site with zero Hamiltonian and loss rate $\gamma_x$. Since $\gamma_x\ge\gamma_{\min}$, this gives $\tr[(N_x+\1)^kT_u^{\mathrm{loss}}(\rho)]\le B_k^{\mathrm{loss}}(u)$. Cauchy--Schwarz inequality and $(N_x+\1)N_y\le((N_x+\1)^2+(N_y+\1)^2)/2$ then give
		\begin{equation*}
		\begin{aligned}
			\|[H,T_u^{\mathrm{loss}}(\rho)]\|_1
			&\le\left(4J_*|E|+2\sum_{x\in V}|\mu_x|\right)\sqrt{B_2^{\mathrm{loss}}(u)}
			+\sum_{x\in V}|U_x|\sqrt{B_4^{\mathrm{loss}}(u)}.
		\end{aligned}
		\end{equation*}
		Duhamel's formula as in Equation~\eqref{eq:adaptive-local-duhamel} and contractivity give
		\begin{equation*}
		\begin{aligned}
			\|T_\delta(\rho)-T_\delta^{\mathrm{loss}}(\rho)\|_1
			&\le\int_0^\delta\|[H,T_u^{\mathrm{loss}}(\rho)]\|_1\,du
			\le\varepsilon_{\mathrm{prep}}(\delta).
		\end{aligned}
		\end{equation*}
		The dependence on $\delta$ follows from $\sqrt{x+y}\le\sqrt x+\sqrt y$, which gives
		\begin{equation*}
			\int_0^\delta\sqrt{B_j^{\mathrm{loss}}(u)}\,du
			\le\frac{2(\ell-1)}{2(\ell-1)-j}
			\left(\frac{2j}{(\ell-1)\gamma_{\min}\ell}\right)^{j/(2(\ell-1))}
			\delta^{1-j/(2(\ell-1))}+[\ell(\ell+1)]^{j/2}\delta,\qquad j=2,4.
		\end{equation*}
		Thus $\varepsilon_{\mathrm{prep}}(\delta)=\mathcal O(\delta^{(\ell-3)/(\ell-1)})$ as $\delta\downarrow0$. If all $U_x=0$, the exponent improves to $\frac{\ell-2}{\ell-1}$. Since $\ell>3$, the bound is finite and uniform in the input and extends in trace norm to arbitrary $\rho$.

		We next bound the error of the occupation test and vacuum reset. Put $\sigma=T_\delta^{\mathrm{loss}}(\rho)$. Markov's inequality and a union bound give, for every $k\ge1$,
		\begin{equation*}
			\tr[(\1-P_V^{(M)})\sigma]
			\le\sum_{x\in V}\tr[(\1-P_x^{(M)})\sigma]
			\le\frac{|V|B_k^{\mathrm{loss}}(\delta)}{(M+2)^k}.
		\end{equation*}
		Optimizing over $k$ and using that probabilities are at most one bounds this trace by $q_{\delta,M}$. The reset term is positive and has trace equal to this failure probability, also in the presence of reference space. The above estimate and $2\sqrt q+q\le3\sqrt q$ for $0\le q\le1$ therefore give, for every $k\ge1$,
		\begin{equation*}
		\begin{aligned}
			\|T_\delta^{\mathrm{loss}}(\rho)-\mathcal E_{\delta,M}(\rho)\|_1
			&\le\|\sigma-P_V^{(M)}\sigma P_V^{(M)}\|_1
			+\tr[(\1-P_V^{(M)})\sigma]\\
			&\le2\sqrt{q_{\delta,M}}+q_{\delta,M}\\
			&\le\frac{3\sqrt{|V|}}{(M+2)^{k/2}}
			\left[\left(\frac{2k}{(\ell-1)\gamma_{\min}\ell\delta}\right)^{k/(2(\ell-1))}
			+[\ell(\ell+1)]^{k/2}\right].
		\end{aligned}
		\end{equation*}
		Finally, taking the supremum over joint states and finite references and applying the triangle inequality gives
		\begin{equation*}
		\begin{aligned}
			\|T_\delta-\mathcal E_{\delta,M}\|_\diamond
			&\le\|T_\delta-T_\delta^{\mathrm{loss}}\|_\diamond
			+\|T_\delta^{\mathrm{loss}}-\mathcal E_{\delta,M}\|_\diamond\\
			&\le\varepsilon_{\mathrm{prep}}(\delta)+2\sqrt{q_{\delta,M}}+q_{\delta,M},
		\end{aligned}
		\end{equation*}
		which is Equation~\eqref{eq:global-preparation-error}.
	\end{proof}

	The on-site Hamiltonian preserves occupation, the jumps lower it, and the truncated hopping preserves $\cH_M$. Thus $\widehat{\cL}_M\coloneqq\cL_{V,M}|_{\cT_1(\cH_M)}$ is a finite-dimensional Lindblad generator, with Hamiltonian $P_V^{(M)}HP_V^{(M)}$ and jumps $\sqrt{\gamma_x}P_x^{(M)}a_x^\ell P_x^{(M)}$. Choose $b=\lceil\log_2(M+1)\rceil$, so that $M<2^b$. The binary expansion of each $0\le n\le M$ defines its encoding into $b$ qubits per site:
	\begin{equation*}
		n=\sum_{j=0}^{b-1}2^j\nu_j,\qquad \nu_j\in\{0,1\},\qquad
		\ket n\longmapsto\ket{\nu_{b-1}\cdots\nu_0}.
	\end{equation*}
	Strings encoding occupations above $M$ are invalid. Each Hamiltonian term and jump vanishes on invalid strings on its own support, so the encoded subspace is invariant. Decoding reverses the encoding on valid strings and maps invalid outputs to the vacuum.

	\begin{thm}[Finite-register simulation of arbitrary inputs]\label{thm:global-finite-register}
		Let $0<\delta<t$, $p>2$, $0<\eta\le1/2$, and let $M\ge2$ be an integer. Suppose the decoded circuit channel $\Psi_{t-\delta}$ approximates $e^{(t-\delta)\widehat{\cL}_M}$ to diamond error $\eta$. Then
		\begin{equation}\label{eq:global-bh-unrestricted-error}
		\begin{aligned}
			\|T_t-\Psi_{t-\delta}\mathcal E_{\delta,M}\|_\diamond
			&\le\varepsilon_{\mathrm{prep}}(\delta)+2\sqrt{q_{\delta,M}}+q_{\delta,M}+\eta+2C_pJ_*|E|M^{-(p-2)/2}
			\int_\delta^t b_p(u)^{p/2}\,du,
		\end{aligned}
		\end{equation}
		where $C_p\le4$ is given in Lemma~\ref{lem:hopping-truncation-error} and $b_p$ in \eqref{eq:aux-bp-fct}. For $0<\varepsilon_\diamond\le1$, the error is at most $\varepsilon_\diamond$ if
		\begin{equation}\label{eq:global-simulation-accuracy-choices}
		\begin{aligned}
			0<\delta&\le\min\left\{\frac t2,
				\left(\frac{\varepsilon_\diamond}{4(1+\varepsilon_{\mathrm{prep}}(1))}\right)^{(\ell-1)/(\ell-3)}\right\},
			\qquad 0<\eta\le\frac{\varepsilon_\diamond}{4},\\
			M&\ge\max\Biggl\{\biggl(\frac{144|V|B_p^{\mathrm{loss}}(\delta)}{\varepsilon_\diamond^2}\biggr)^{1/p}-2,\biggl(\frac{8C_pJ_*|E|}{\varepsilon_\diamond}
			\int_\delta^t b_p(u)^{p/2}\,du\biggr)^{2/(p-2)}\Biggr\}.
			\end{aligned}
		\end{equation}
		In particular, choosing $p=2(\ell-1)$ makes the dependence on $t$, $\delta$, and $M$ explicit: for $0<\delta\le1$ and $\delta<t$,
		\begin{equation*}
		\begin{aligned}
			\|T_t-\Psi_{t-\delta}\mathcal E_{\delta,M}\|_\diamond
			&\le\eta+\varepsilon_{\mathrm{prep}}(1)\delta^{(\ell-3)/(\ell-1)}+\mathcal O\!\left(\frac{\sqrt{|V|}}{\delta M^{\ell-1}}
			+\frac{J_*|E|}{M^{\ell-2}}\left[t-\delta+\log\frac t\delta\right]\right).
		\end{aligned}
		\end{equation*}
		The implicit constant depends only on $\ell$, $\gamma_{\min}$, $\gamma_{\max}$, $J_*$, $\Delta_G$, and the fixed Sobolev weight.

		For fixed $t$, $p$, $\varepsilon_\diamond$, graph, and model parameters, the cutoff conditions permit $M=\mathcal O(\delta^{-1/(\ell-1)})$ as $\delta\downarrow0$. Thus at most $\log_2(1/\delta)/(\ell-1)+\mathcal O(1)$ data qubits per site suffice.
	\end{thm}
	\begin{proof}
		Figure~\ref{fig:global-simulation-comparison} summarizes the three approximations. The Duhamel identity \eqref{eq:adaptive-local-duhamel}, with $R=V$ and constant cutoff $M$, gives on $[\delta,t]$
		\begin{equation*}
			T_t-e^{(t-\delta)\cL_{V,M}}T_\delta
			=\int_\delta^t e^{(t-u)\cL_{V,M}}
			(\cL-\cL_{V,M})T_u\,du.
		\end{equation*}
		Lemma~\ref{lem:hopping-truncation-error} and Proposition~\ref{main-prop:sobolev-regularizing} bound its diamond norm by the integral term in Equation~\eqref{eq:global-bh-unrestricted-error}. As in Corollary~\ref{cor:local-channel-approximation}, the moment estimates are uniform in a finite reference. Since $\delta>0$, the integral is finite for every $p>2$. Now insert $e^{(t-\delta)\cL_{V,M}}\mathcal E_{\delta,M}$. Contractivity and Proposition~\ref{prop:global-onsite-preparation} bound the preparation error, while invariance of $\cH_M$ identifies this channel with $e^{(t-\delta)\widehat{\cL}_M}\mathcal E_{\delta,M}$. The circuit error adds at most $\eta$.

		For $0<\delta\le1$, changing variables in Equation~\eqref{eq:global-preparation-quantities} gives $\varepsilon_{\mathrm{prep}}(\delta)\le\delta^{(\ell-3)/(\ell-1)}\varepsilon_{\mathrm{prep}}(1)$. The choices in Equation~\eqref{eq:global-simulation-accuracy-choices} therefore make the preparation error at most $\varepsilon_\diamond/4$ and give $q_{\delta,M}\le(\varepsilon_\diamond/12)^2$. Since $2\sqrt q+q\le3\sqrt q$ for $0\le q\le1$, each of the projection, hopping, and circuit errors is also at most $\varepsilon_\diamond/4$.

		To make the time dependence of the hopping error explicit, Equation~\eqref{eq:aux-bp-fct} gives
		\begin{equation*}
			b_p(u)^{p/2}\le2^{p/2-1}\sqrt{Z_\kappa^*}\left[
				\left(\frac{2p}{(\ell-1)\gamma_{\min}\ell u}\right)^{p/(2(\ell-1))}
				+\Delta_{\ell,p}^{p/2}\right].
		\end{equation*}
		The second term in brackets integrates to $\Delta_{\ell,p}^{p/2}(t-\delta)$, while
		\begin{equation*}
			\int_\delta^t u^{-p/(2(\ell-1))}\,du
			=\begin{cases}
			\displaystyle\frac{t^{1-p/(2(\ell-1))}-\delta^{1-p/(2(\ell-1))}}
			{1-p/(2(\ell-1))},&p\ne2(\ell-1),\\[6pt]
			\log(t/\delta),&p=2(\ell-1).
			\end{cases}
		\end{equation*}
		For $p=2(\ell-1)$, this gives the stated asymptotic error bound, since $B_p^{\mathrm{loss}}(\delta)^{1/2}=\mathcal O(\delta^{-1})$ for $0<\delta\le1$. For general $p>2$, these estimates and $B_p^{\mathrm{loss}}(\delta)^{1/p}=\mathcal O(\delta^{-1/(\ell-1)})$ show that both cutoff conditions admit $M=\mathcal O(\delta^{-1/(\ell-1)})$ for fixed parameters. Taking the logarithm gives the stated qubit count.
	\end{proof}

	We now construct the circuit required in Theorem~\ref{thm:global-finite-register} and bound its resources using \cite[Thm.~1.2]{ChenGaoWangZhou.2026}. A circuit consists of one- and two-qubit gates, with auxiliary qubits that may be initialized and discarded. We allow gates between any pair of qubits.
    
    We first describe the unitary encodings needed by the algorithm. A \emph{block encoding} represents an operator, divided by a normalization factor, as a block of a unitary acting on the system and auxiliary qubits. The combined jump map $\sum_{x\in V}\ket x\otimes\sqrt{\gamma_x}P_x^{(M)}a_x^\ell P_x^{(M)}$ is encoded in a rectangular block, with an extra output register recording the site $x$. The cutoff Hamiltonian $P_V^{(M)}HP_V^{(M)}$ on $\cH_M$ and the above map admit encodings with normalization factors $\alpha_H$ and $\alpha_B$, where
    \begin{equation}\label{eq:global-encoding-normalizations}
        \begin{aligned}
            \alpha_H&=1+\sum_{x\in V}\max_{n=0,\ldots,M}
            \left|\frac{U_x}{2}n(n-1)-\mu_xn\right|+2J_*|E|M,\\
            \alpha_B^2&=1+M^\ell\sum_{x\in V}\gamma_x.
        \end{aligned}
    \end{equation}
    To construct these encodings, we use the action of the local terms in the occupation basis. The on-site Hamiltonian is diagonal, hopping changes two occupations, and loss changes one. A site or edge register controls which local operator is encoded. Reversible arithmetic computes its matrix elements, and controlled rotations and occupation shifts implement the encoding. The individual jump encodings are combined as in \cite[Rem.~1.3]{ChenGaoWangZhou.2026}. The Hamiltonian encoding unitary can be made Hermitian using one additional auxiliary qubit, as required by the theorem \cite[Sec.~1.1]{ChenGaoWangZhou.2026}.
    
    We next count the calls and gates used by the simulation algorithm. With $\tau=(t-\delta)(\alpha_H+\alpha_B^2)$, the cited theorem, applied with diamond error $\eta/2$, uses
    \begin{equation}\label{eq:chen-lindblad-queries}
        N_{\mathrm{qry}}=\mathcal O\!\left(\tau+
        \frac{\log(1/\eta)}{\log(e+\log(1/\eta)/\tau)}\right)
    \end{equation}
    calls to the encoding unitaries, their inverses, or controlled versions. If each call uses at most $g_{\mathrm{enc}}$ gates and the two encodings together use $a_{\mathrm{anc}}$ auxiliary qubits, the total gate count is
    \begin{equation}\label{eq:chen-lindblad-gates}
        \mathcal O\!\left(N_{\mathrm{qry}}\left[g_{\mathrm{enc}}+a_{\mathrm{anc}}
        +\log\!\left(\frac{e(1+\tau)}{\eta}\right)
        \left(\log(|V|+1)+\log\!\left(\frac{e(1+\tau)}{\eta}\right)\right)\right]\right).
    \end{equation}

    It remains to account for the cost and precision of the encodings. We approximate each encoding unitary in norm to accuracy $\eta/[4(1+N_{\mathrm{qry}})]$. For coefficients supplied classically to the required precision, selecting local terms and computing their amplitudes reversibly gives
    \begin{equation*}
        g_{\mathrm{enc}}+a_{\mathrm{anc}}
        =\mathcal O\!\left((|V|+|E|)\operatorname{poly}\!\left(
        \ell,\log(M+1),\log\!\frac{e(|V|+|E|)(1+N_{\mathrm{qry}})}{\eta}\right)\right).
    \end{equation*}
    A telescoping sum bounds the additional diamond error by $2N_{\mathrm{qry}}\eta/[4(1+N_{\mathrm{qry}})]\le\eta/2$. Together with the algorithmic error, this gives total diamond error at most $\eta$. Decoding does not increase this error and yields the channel $\Psi_{t-\delta}$ required in Theorem~\ref{thm:global-finite-register}.

    The circuit uses $|V|\lceil\log_2(M+1)\rceil$ data qubits, in addition to auxiliary qubits. The gate count covers the qubit simulation and excludes physical loss preparation, occupation tests, and coherent transfer between bosons and qubits.

    Finally, locality reduces the cost when only the output on $X$ is needed. Apply the same construction to $R=X[r]$ and add the locality error from Corollary~\ref{cor:local-channel-approximation}. All graph sizes and coefficient sums are then restricted to $R$. Since $|E(R)|\le\Delta_G|R|/2$, the cost depends on the simulated region and its coefficients, independently of the ambient volume.

\section{Discussion}\label{sec:discussion}
    We show that multiphoton loss restores operator-norm locality for the dissipative Bose--Hubbard model with $\ell>2$, uniformly over all input states. Moment regularization supplies the occupation bounds, while an adaptive cutoff controls their singularity at time zero. Together with a bounded-interaction Lieb--Robinson estimate insensitive to on-site terms, this gives local channel approximation and a quantitative thermodynamic limit without assumptions on the initial moments.
    
    The extension to quartic interactions with shifted four-photon loss gives a local adiabatic approximation for states in the cat-code space, uniformly in the volume. For $\ell>3$, physical loss preparation and an occupation test also yield a hybrid bosonic--qubit simulation scheme for arbitrary inputs, with explicit diamond-norm error bounds.

    An important open question is whether two-photon loss also yields locality uniformly over all input states. A possible approach is a short-time comparison that controls hopping and dissipative damping together. Such an estimate would need to overcome the nonintegrable singularity produced by the present moment bounds.

    Our bounds give almost-ballistic propagation at fixed accuracy, with radii growing as $t^{1+\eta}$ for every $\eta>0$. A finite-velocity exponential light cone remains open, since the constants grow with the moment order and do not permit taking $p\to\infty$ directly.

\vspace{2ex}
\emph{Acknowledgments.} We would like to thank Frank Verstraete for valuable feedback on the topic.

\section*{Declarations}

\emph{Funding.} This work was funded by the ERC Starting Grant ``MathQuantProp'' No.\ 101163620 (ML)\footnote{Views and opinions expressed are however those of the authors only and do not necessarily reflect those of the European Union or the European Research Council Executive Agency. Neither the European Union nor the granting authority can be held responsible for them}. Additional funding came from the Deutsche Forschungsgemeinschaft (DFG, German Research Foundation), Project-ID 470903074, TRR 352 (ML, TM), Project-ID 465199066 (ML), and Project-ID 575156903 (TM). The QuantERA II Programme received funding from the EU's H2020 research and innovation programme under GA No.\ 101017733 (TM).

\emph{Competing interests.} The authors declare that they have no competing interests relevant to this article.

\emph{Data availability.} No datasets were generated or analysed during this study.

\emph{AI declaration.}
The idea for this paper and the main steps of the proof of the central result were developed in discussions between the authors. OpenAI's Sol and Astra models, accessed through ChatGPT/Codex, assisted with finalizing proof details, proofreading, improving the clarity and consistency of the writing, and producing TikZ code for figures illustrating the authors' arguments. The authors verified the code and figures and take full responsibility for the mathematical arguments and the final manuscript.

\sloppy
\bibliographystyle{IEEEtran}
\bibliography{bibliography_v4}

\appendix
\addtocontents{toc}{\protect\setcounter{tocdepth}{1}}
\addcontentsline{toc}{section}{Appendices}
\addtocontents{toc}{\protect\setcounter{tocdepth}{0}}

\section{Sobolev Regularization and Moment Propagation}\label{sec:sobolev}
	\begin{prop}[Dissipative Bose--Hubbard]\label{prop:sobolev-regularizing}
		Let $(\cL,\cT_f)$ be the dissipative Bose--Hubbard model \eqref{eq:lindblad} on a lattice $(V,E)$ of uniformly bounded maximum degree $\Delta_G$ and hopping coefficients $J_*=\sup_{(i,j)\in E}|J_{ij}|<\infty$ for any $\ell\ge2$ and loss rates $0<\gamma_{\min}\le\gamma_j\le\gamma_{\max}<\infty$. Then the closure of $(\cL,\cT_f)$ generates a QMS on $\cT_1$, which is Sobolev preserving and regularizing in the senses of Definitions~\ref{def:multimode-sobolev-preserving-qms} and~\ref{def:multimode-sobolev-regularizing-qms}. For every $k\ge1$, $v\in V$, and $t>0$,
		\begin{equation*}
			\|e^{t\cL}\|_{1\to W_v^{k,1}}\le\left(\frac{2k}{(\ell-1)\gamma_{\min}\ell t}\right)^{k/(\ell-1)}+\Delta_{\ell,k}^k,
		\end{equation*}
		where
		\begin{equation*}
			\Delta_{\ell,k}=\frac{2}{\gamma_{\min}\ell}\left(\gamma_{\max}\frac{(\ell+1)\ell^2}{2}+2^kkJ_*\Delta_G(1+e^\kappa)\right).
		\end{equation*}
		The constants are independent of the finite volume and center $v$.
	\end{prop}
	\begin{proof}
		First, the bound on the Hamiltonian part is given in Lemma~\ref{lem-aux:bose-hubbard-moment-propagation} for any $w\in V$ by
		\begin{equation*}
			i[H,f(N_w)]\leq 2^kkJ_*\sum_{j:(j,w)\in E}\left((N_w+\1)^{k}+(N_j+\1)^k\right)\,.
		\end{equation*}

		Second, the bound on the on-site photon-loss is given in Lemma~\ref{lem-aux:l-diss-moment-propagation} by
		\begin{equation*}
			\tr\big[\cL[a^\ell](\rho)(f(N))\big]\le -\ell\tr\big[\rho\,(N+\1)^{\ell+k-1}\,\big]+\frac{(\ell+1)\ell^2}{2}\tr\big[\rho\,(N+\1 )^{\ell+k-2}\big]\,.
		\end{equation*}

		Combining both bounds with respect to the concentrated Sobolev norm at $v$ gives the following weighted estimate. For neighboring sites $w\sim j$, the weights $q_z=e^{-\kappa\operatorname{dist}(v,z)}/Z_v$ satisfy $q_w\le e^\kappa q_j$ and $q_j\le e^\kappa q_w$, while each site has at most $\Delta_G$ neighbors. Hence the full generator satisfies
		\begin{equation*}
			\begin{aligned}
				\tr[\cW_v^k(\cL(\rho))]&=\sum_{w\in V} \frac{e^{-\kappa \, \mathrm{dist}(v,w)}}{Z_v}\tr\bigl[\cL(\rho)(N_w+\1)^{k}\bigr]\\
				&\leq\sum_{w\in V} \frac{e^{-\kappa \, \mathrm{dist}(v,w)}}{Z_v}\Bigl(-\gamma_{\min}\ell\tr\big[\rho\,(N_w+\1)^{\ell+k-1}\big]\\
				&\qquad\qquad\qquad\qquad\quad + \gamma_{\max}\frac{(\ell+1)\ell^2}{2}\tr\big[\rho\,(N_w+\1 )^{\ell+k-2}\big]\\
				&\qquad\qquad\qquad\qquad\quad +2^kkJ_*\Delta_G(1+e^\kappa)\tr[\rho(N_w+\1)^{k}]\Bigr)\\
				&\leq\sum_{w\in V} \frac{e^{-\kappa \, \mathrm{dist}(v,w)}}{Z_v}\Bigl(-\frac{\gamma_{\min}\ell}{2}\tr\big[\rho\,(N_w+\1)^{\ell+k-1}\big]-\frac{\gamma_{\min}\ell}{2}\tr\big[\rho\,(N_w+\1)^{\ell+k-1}\big] \\
				&\qquad\qquad\qquad\qquad\quad + \Bigl(\gamma_{\max}\frac{(\ell+1)\ell^2}{2}+2^kkJ_*\Delta_G(1+e^\kappa)\Bigr)\tr\big[\rho\,(N_w+\1 )^{\ell+k-2}\big]\Bigr)
			\end{aligned}
		\end{equation*}
		for all positive finite-rank states $\rho\in\cT_f$ with $\tr\rho=1$. Applying the scalar optimization on the spectrum of the number operators as in \cite[Lem.~S18]{Moebus.2024Learning} gives
		\begin{equation}\label{eq:optimization}
			\sup_{x\geq0}\left(-x^\nu+\Delta_{\ell,k} x^{\nu-1}\right)=\Delta_{\ell,k}^\nu\Bigl(\frac{(\nu-1)^{\nu-1}}{\nu^\nu}\Bigr)\leq\Delta_{\ell,k}^\nu
		\end{equation}
		for $\nu=k+\ell-1\geq1$ and $\Delta_{\ell,k} = \frac{2}{\ell\gamma_{\min}}\bigl(\gamma_{\max}\frac{(\ell+1)\ell^2}{2}+2^kkJ_*\Delta_G(1+e^\kappa)\bigr)$, which shows
		\begin{equation*}
			\begin{aligned}
				\tr[\cW_v^k(\cL(\rho))]&=\sum_{w\in V} \frac{e^{-\kappa \, \mathrm{dist}(v,w)}}{Z_v}\tr\bigl[\cL(\rho)(N_w+\1)^{k}\bigr]\\
				&\leq-\frac{\gamma_{\min}\ell}{2}\|\rho\|_{W^{k+\ell-1,1}_v}
				+\frac{\gamma_{\min}\ell}{2}\Delta_{\ell,k}^\nu
			\end{aligned}
		\end{equation*}
		These estimates hold for every $k\geq1$ for the considered generator and, by homogeneity, for every positive $\rho\in\cT_f$ with its constant term multiplied by $\tr\rho$. Defining
		\begin{equation*}
			\delta=\ell-1,\qquad \mu_k=\frac{\gamma_{\min}\ell}{2},\qquad
			c_k=\frac{\gamma_{\min}\ell}{2}\Delta_{\ell,k}^{k+\ell-1},
		\end{equation*}
		it gives
		\begin{equation*}
			\tr[M_{v,k}\cL(\rho)]
			\leq-\mu_k\tr[M_{v,k+\delta}\rho]+c_k\tr\rho
			\leq c_k\tr[M_{v,k}\rho],
		\end{equation*}
		since $M_{v,k}\geq\1$. The first inequality is the concentrated Sobolev regularization condition of \cite[Assumption~3.2]{Rico.2025} with $\delta=\ell-1$. The second is the concentrated Sobolev stability condition of \cite[Assumption~2.2]{Rico.2025}. The Hamiltonian is a symmetric polynomial on $\cH_f$, and the jump operators are polynomials. Thus \cite[Cor.~3.8]{Rico.2025}, applied with the unbounded set of orders $k\geq1$, generates the conservative QMS and its strongly continuous restrictions to $W_v^{k,1}$, with $\cT_f$ as a generator core and
		\begin{equation*}
			\|e^{t\cL}\|_{W_v^{k,1}\to W_v^{k,1}}\leq e^{c_kt}.
		\end{equation*}
		For a bounded degree $d=\max\{4,\ell\}$, \cite[Lem.~D.2]{Rico.2025} gives $W_v^{r,1}\subseteq\cD(\cL_{v,k})$ for $r\geq k+2d$, where $\cL_{v,k}$ denotes the weighted generator. Moreover, \cite[Thm.~3.12]{Rico.2025} shows Sobolev regularization with respect to the concentrated Sobolev norm:
		\begin{equation*}
			\| e^{t \cL} \|_{1 \to W^{k,1}_v} \le \left( \frac{2k}{(\ell-1)\gamma_{\min}\ell t} \right)^{\frac{k}{\ell-1}} +  \Delta_{\ell,k}^k\,.
		\end{equation*}
		The constants in the preservation and regularization bounds are independent of the finite volume and the center when the displayed local coefficient and degree bounds are uniform.
	\end{proof}

	The auxiliary moment bounds used above are the following.

	\begin{lem}\label{lem-aux:bose-hubbard-moment-propagation}
		Let $(H,\cH_f)$ be the Bose--Hubbard Hamiltonian on a lattice $(V,E)$ defined in Equation~\eqref{eq:bose-hubbard}. Then, for every $k\ge1$,
		\begin{align*}
			\begin{aligned}
				i[H,f(N_w)]&\leq 2^kk J_*\sum_{j:(j,w)\in E}\left((N_w+1)^{k}+(N_j+1)^k\right)\,,
			\end{aligned}
		\end{align*}
		as a quadratic-form inequality on $\cH_f$, for all $w\in V$ and $f(x)=(x+1)^{k}1_{x\ge -1}$, with $J_*=\sup_{(i,j)\in E}|J_{ij}|$.
	\end{lem}
		\begin{proof}
			For the function $f(x)=(x+1)^{k}1_{x\ge -1}$, we use $a_wf(N_w)=f(N_w + \1)a_w$ to achieve
			\begin{align*}
				i[H,f(N_w)]&=-i\sum_{\substack{i,j\in V\\(i,j)\in E}}J_{ij}\left[a_i^\dagger\,a_j,\,f(N_w)\right]\\
				&=-i\sum_{j:(j,w)\in E}\left(J_{jw}a_wg(N_w)a^\dagger_j-J_{wj}a_jg(N_w)a^\dagger_w\right)
			\end{align*}
			for $g(x)=f(x)-f(x-1)$ and $J_{jw}=0$ for $(j,w)\notin E$. Next, we express the above discrete operator in the Fock basis and represent it using matrices of rank two:
			\begin{align*}
				i[H,f(N_w)]&=-i\sum_{j:(j,w)\in E}\sum_{n_w,n_j=0}^{\infty}J_{jw}g(n_w)a_w\ket{n_w}\bra{n_w}\otimes\ket{n_j}\bra{n_j}a^\dagger_j+ h.c.\\
				&=-i\sum_{j:(j,w)\in E}\sum_{n_w,n_j=1}^{\infty}J_{jw}g(n_w)\sqrt{n_w n_j}\,\ket{n_w-1,n_j}\bra{n_w,n_j-1}+ h.c.
			\end{align*}
			Interpreting each summand with $n_w,n_j\ge1$ as a $2\times2$ matrix defined by
			\begin{equation*}
				\begin{aligned}
					\alpha\ket{e_1}\bra{e_2}+\overline{\alpha}\ket{e_2}\bra{e_1}&=-iJ_{jw}g(n_w)\sqrt{n_w n_j}\ket{n_w-1,n_j}\bra{n_w,n_j-1}\\
					&\qquad+i\overline{J_{jw}}g(n_w)\sqrt{n_w n_j}\ket{n_w,n_j-1}\bra{n_w-1,n_j}\\
					&\leq |\alpha|(\ket{e_1}\bra{e_1}+\ket{e_2}\bra{e_2}).
				\end{aligned}
			\end{equation*}
			The last inequality follows from the eigenvalue decomposition, whose eigenvalues are $\pm|\alpha|$. Then,
			\begin{equation*}
				\begin{aligned}
					i[H,f(N_w)]&\leq\sum_{j:(j,w)\in E}\sum_{n_w,n_j=1}^{\infty}|J_{jw}|g(n_w)\sqrt{n_w n_j}\,\bigl(\ket{n_w\!-\!1,n_j}\bra{n_w\!-\!1,n_j}+\ket{n_w,n_j\!-\!1}\bra{n_w,n_j\!-\!1}\bigr)\\
					&\leq 2\sum_{j:(j,w)\in E}\sum_{n_w,n_j=0}^{\infty}|J_{jw}|g(n_w+1)\sqrt{(n_w+1)(n_j+1)}\ket{n_w,n_j}\bra{n_w,n_j}\\
					&\leq 2 k J_*\sum_{j:(j,w)\in E}\sum_{n_w,n_j=0}^{\infty}(n_w+2)^{k-1}\sqrt{(n_w+1)(n_j+1)}\ket{n_w,n_j}\bra{n_w,n_j}\,,
				\end{aligned}
			\end{equation*}
			using $g(x)\le k(x+1)^{k-1}$ for $x\ge0$ (see \cite[Lem.~C.2]{Gondolf.2024}). Young's inequality, i.e.
			\begin{equation*}
				ab\leq \frac{a^p}{p}+\frac{b^q}{q}\qquad\text{for all }a,b\geq0\,,
			\end{equation*}
			with conjugate exponents $p=\frac{2k}{2k-1}$ and $q=2k$, shows
			\begin{align*}
				\begin{aligned}
					i[H,f(N_w)]&\leq J_* 2^{k-1}\sum_{j:(j,w)\in E}\sum_{n_w,n_j=0}^{\infty}\left((2k-1)(n_w+1)^{k}+(n_j+1)^k\right)\ket{n_w,n_j}\bra{n_w,n_j}\\
					&{\leq} 2^kk J_*\sum_{j:(j,w)\in E}\left((N_w+1)^{k}+(N_j+1)^k\right)\,.
				\end{aligned}
			\end{align*}
			This concludes the first auxiliary result.
		\end{proof}

	\begin{lem}\label{lem-aux:l-diss-moment-propagation}
			Let $\cL[a^\ell]$ be the single-site $\ell$-photon dissipation term appearing in Equation~\eqref{eq:lindblad}, and let $k\ge 1$ and $\ell\ge 2$. Then, for any state $\rho\in\cT_f$,
			\begin{equation*}
				\tr\big[\cL[a^\ell](\rho)(f(N))\big]\le -\ell\tr\big[\rho\,(N+\1)^{\ell+k-1}\,\big]+\frac{(\ell+1)\ell^2}{2}\tr\big[\rho\,(N+\1 )^{\ell+k-2}\big]\,,
			\end{equation*}
			where $f(x)=(x+1)^{k}1_{x\ge -1}$.
		\end{lem}
		\begin{proof}
			For the function $f(x)=(x+1)^{k}1_{x\ge -1}$, we use $af(N+j\1)=f(N + (j+1)\1)a$ to achieve
			\begin{equation*}
				\begin{aligned}
					\cL[a^\ell]^\dagger(f(N))&=(a^\dagger)^\ell f(N)a^\ell-\frac{1}{2}\Big((a^\dagger)^\ell a^\ell f(N)+f(N)(a^\dagger)^\ell a^\ell\Big)=-(a^\dagger)^\ell a^\ell g_\ell(N)
				\end{aligned}
			\end{equation*}
			for $g_\ell(x)=f(x) - f(x - \ell)$. Next, we use $(a^\dagger)^\ell a^\ell=(N-(\ell-1)\1)\cdots(N-\1)N\eqqcolon N[-\ell+1:0]$ with $N[r:j]\coloneqq (N+r\1)\cdots (N+j\1)$ and the convention $N[r:j]=\1$ whenever $r>j$. Therefore,
			\begin{align*}
				\tr\big[\rho\,\cL[a^\ell]^\dagger(f(N))\big]&=-\tr\big[\rho\,N[-\ell+1:0]g_\ell(N)\big]\,.
			\end{align*}
			We lower bound the last trace as follows:
			\begin{equation*}
				\begin{aligned}
					\tr\big[\rho\,N[-\ell+1:0]\,g_\ell(N)\big]&\overset{(1)}{\geq}\ell\tr\big[\rho\,N[-\ell+1:0](N+\1)^{k-1}\big]\\
					&\overset{(2)}{\geq}\ell\,\tr\big[\rho\,(N+\1)^{\ell+k-1}\,\big]-\frac{(\ell+1)\ell^2}{2}\tr\big[\rho\,(N+\1 )^{\ell+k-2}\big]\,.
				\end{aligned}
			\end{equation*}
			In $(1)$, we used $g_\ell(x)\geq(x+1)^{k-1}\ell$ for $x>\ell-1$, which follows by
			\begin{equation*}
				\begin{aligned}
					(x+1)^{-(k-1)}g_\ell(x)&=k\ell\int_{0}^1\left(1-s\frac{\ell}{x+1}\right)^{k-1}ds\geq k\ell\int_{0}^1\left(1-s\right)^{k-1}ds\geq\ell\,.
				\end{aligned}
			\end{equation*}
			In $(2)$, we used that for
			\begin{align*}
				N[-\ell+1:0]&=\sum_{n\ge \ell}\,(n-\ell+1)\dots n\,\ketbra{n}{n}\\
				&\overset{(3)}{\ge} \sum_{n\ge \ell}\left((n+1)^\ell-\frac{(\ell+1)\ell}{2}(n+1)^{\ell-1}\right)\ketbra{n}{n}\\
				&\ge (N+\1)^\ell-\frac{(\ell+1)\ell}{2}\,(N+\1)^{\ell-1}\,,
			\end{align*}
			where $(3)$ comes from Lemma~\ref{lem:bounds-ccr-l-product} below, and the last inequality uses $(n+1)\leq\frac{(\ell+1)\ell}{2}$ for $n<\ell$ and $2\leq\ell$. Thus,
			\begin{equation}\label{eq-ex:l-dissipation-upper-bound}
				\begin{aligned}
					\tr\big[\cL[a^\ell](\rho)(f(N))\big]&\le -\ell\tr\big[\rho\,(N+\1)^{\ell+k-1}\,\big]+\frac{(\ell+1)\ell^2}{2}\tr\big[\rho\,(N+\1 )^{\ell+k-2}\big]\,.
				\end{aligned}
			\end{equation}
		\end{proof}
		Finally, we recall the technical lemma used in the above proof:
	\begin{lem}\label{lem:bounds-ccr-l-product}
			If $\ell\in\N$ and $x\geq \ell$, then
			\begin{equation*}
				\begin{aligned}
					(x+1)^\ell-\frac{(\ell+1)\ell}{2}(x+1)^{\ell-1}
					\leq ((x+1)-\ell)\cdots ((x+1)-1)
					\leq (x+1)^\ell .
				\end{aligned}
			\end{equation*}
	\end{lem}
	\begin{proof}
			See \cite[Lem.~C.3]{Gondolf.2024}.
	\end{proof}

\section{Bounded-Interaction Dissipative Lieb--Robinson Bounds}\label{sec:lrb}
	We derive the bounded-interaction estimate used in Proposition~\ref{prop:bounded-interaction-lrb} and the adaptive locality proofs from \cite[Thm.~3.1]{Rubiliani.2026}. We retain the time dependence of the edge interactions and then remove the on-site cutoff.

	\begin{cor}[Nearest-neighbor bounded-interaction bound]\label{cor:lrb-bounded-interaction-graph}
		Let $(V,E)$ be a finite graph with maximum degree $\Delta_G$ and graph distance $d$, and fix $T>0$. Let $S_t$ be a support-preserving Heisenberg adjoint of a contractive $C_0$-semigroup on $\cT_1$. Assume that $S_t$ admits adjoint approximations $S_t^{(L)}$ generated by sums of bounded normal contraction generators supported on individual sites, whose preduals converge strongly on $\cT_1$, locally uniformly in time, to the predual of $S_t$.

		For $u\in[0,T]$, let $\Phi_e(u)$ be bounded normal superoperators supported on edges $e\in E$, each generating a contraction semigroup. Assume that each $\Phi_e$ is piecewise constant with finitely many jumps. Set
		\begin{equation*}
			\|\Phi(u)\|_1\coloneqq\sup_{x\in V}\sum_{\substack{e\in E\\x\in e}}\|\Phi_e(u)\|_{\infty\to\infty},
			\qquad \|\Phi\|_1\coloneqq\sup_{0\le u\le T}\|\Phi(u)\|_1.
		\end{equation*}
		Let $\tau_{s,t}$ be the Heisenberg propagator obtained by adding these edge terms to the on-site dynamics. Then, for $0\le s\le t\le T$, every bounded observable $A$ supported in $X$, and every bounded normal superoperator $K_Y$ supported in $Y$ with $X\cap Y=\varnothing$,
		\begin{equation}\label{eq:graph-bounded-interaction-lrb}
			\|K_Y\tau_{s,t}(A)\|_\infty
			\le
			\frac{1}{2}\|K_Y\|_{\infty\to\infty}\|A\|_\infty\min\{|X|,|Y|\}\,
			e^{-d(X,Y)}
			\left(e^{2\mathrm e\int_s^t\|\Phi(u)\|_1\,du}-1\right).
		\end{equation}
		For time-independent edge terms, the exponent is $2\mathrm e\|\Phi\|_1(t-s)$, with $\|\Phi\|_1\le\Delta_G\sup_e\|\Phi_e\|_{\infty\to\infty}$. The bound is independent of the on-site generators and the local Hilbert-space dimensions.
	\end{cor}
	\begin{proof}
		We first keep the on-site generators bounded and denote the resulting propagator by $\tau_{s,t}^{(L)}$. The adaptive-Dyson expansion in \cite[Thm.~3.1]{Rubiliani.2026} contains only paths of edge terms, since one-site terms cannot occur in a self-avoiding path. Retaining the instantaneous interaction norm in the path count of \cite[Lem.~4.2 and proof of Thm.~3.5]{Rubiliani.2026} gives
		\begin{equation*}
			\|K_Y\tau_{s,t}^{(L)}(A)\|_\infty
			\le
			\frac{1}{2}\|K_Y\|_{\infty\to\infty}\|A\|_\infty\min\{|X|,|Y|\}\,
			\sum_{n\ge d(X,Y)}\frac{\left(2\int_s^t\|\Phi(u)\|_1\,du\right)^n}{n!}.
		\end{equation*}
		Indeed, paths with $n$ edges require $n\ge d(X,Y)$, and the ordered time integral of $\prod_{j=1}^n\|\Phi(u_j)\|_1$ equals $\bigl(\int_s^t\|\Phi(u)\|_1\,du\bigr)^n/n!$. Multiplying each summand by $e^{n-d(X,Y)}\ge1$ and summing proves Equation~\eqref{eq:graph-bounded-interaction-lrb} for $\tau_{s,t}^{(L)}$. If $d(X,Y)=+\infty$, no connecting path exists and the left side vanishes. Convex interpolation at the finitely many jumps preserves the contraction-generator property. Convergence in the time-integrated norm then extends the bound to the stated schedules by Duhamel's formula. We next remove the on-site cutoff. The bounded-perturbation expansion on each time interval and convergence of finite products give strong convergence of the perturbed predual propagators. Hence $\tau_{s,t}^{(L)}(A)\to\tau_{s,t}(A)$ weak*. Normality of $K_Y$ and weak* lower semicontinuity of the operator norm give Equation~\eqref{eq:graph-bounded-interaction-lrb}.
	\end{proof}

	For both the Bose--Hubbard and cat-code models, the required on-site approximations are obtained by replacing each Hamiltonian and jump by $P_z^{(L)}H_zP_z^{(L)}$ and $P_z^{(L)}L_zP_z^{(L)}$, respectively. The resulting bounded GKSL generators on the full one-mode space eventually agree with the polynomial generator on every finite-Fock matrix unit. These matrix units span a generator core by \cite[Cor.~3.8]{Rico.2025}, so the semigroup approximation theorem \cite[Chap.~III, Sec.~4]{Engel.2000} gives strong trace-class convergence, locally uniformly in time. Finite tensor products preserve this convergence and support preservation.

	For a cutoff schedule $M$, take the on-site dynamics on $R$, extended by the identity outside $R$, and set $\Phi_e(u)=i[H_e^{(M(u))},\cdot]$ on $E(R)$ and zero otherwise. The interaction bounds in the main proofs give
	\begin{equation*}
		\int_s^t\|\Phi(u)\|_1\,du\le
		\begin{cases}
			4\Delta_GJ_*\int_s^t M(u)\,du,&\text{Bose--Hubbard},\\[3pt]
			C\Delta_GJ_{\mathrm{cat}}\int_s^t M(u)^2\,du,&\text{four-photon polynomial model}.
		\end{cases}
	\end{equation*}
	Thus Equation~\eqref{eq:graph-bounded-interaction-lrb}, applied with $K_Y=i[H_e^{(M(s))},\cdot]$ for $e\in\partial_E R$ disjoint from $X$, gives the boundary estimates in the adaptive proofs. Bounded projected code dynamics is included by taking the on-site approximations equal to the original semigroup.

\par
\section{Finite-region adiabatic estimate}\label{app:local-adiabatic-proof}
	\begin{lem}\label{lem:cat-finite-adiabatic}
		Consider the cat model in Section~\ref{subsec:local-adiabatic-limit}. Then the dissipative evolution is approximated by the projected Hamiltonian evolution with error
		\begin{equation}\label{eq:cat-finite-adiabatic}
			\|e^{t\cL_R^\gamma}(\sigma)
				-e^{t\mathcal K_R^{\mathrm{proj}}}(\sigma)\|_1
				\le \frac{C(t|R|^2+|R|)}{\gamma}
		\end{equation}
		for every finite region $R$, $t\ge0$, $\gamma>0$, and state $\sigma=P_R\sigma P_R$. The constant $C<\infty$ depends only on $\sup_x|\alpha_x|$, $J_{\mathrm{cat}}$, the uniform bound on the coefficients of $H_x$, and $\Delta_G$. It is independent of $R$, $t$, and $\gamma$. For every fixed $k\ge1$, $\|(N_x+\1)^kP_x\|_\infty$ is uniformly bounded in $x$.
	\end{lem}
	\begin{proof}
		We first bound moments on the code space. A one-site code state is stationary under $\cL[L_x]$. Apply Equation~\eqref{eq:cat-local-moments} at order $2k$ to the one-site model with zero Hamiltonian and unit loss rate, and let time tend to infinity. Taking the supremum over code states bounds $\|(N_x+\1)^kP_x\|_\infty$ uniformly for bounded $|\alpha_x|$.

		Next, we verify a gap for $L'_R=\sum_{x\in R}L_x^\dagger L_x$ and bound the Hamiltonian relative to it. The gap estimates \cite[Lem.~E.5, E.11]{Moebus.2025} give $L'_R\ge24(\1-P_R)$. The one-site polynomial estimates \cite[Lem.~C.1 and Cor.~C.4]{Moebus.2025} imply, uniformly in $x$,
		\begin{equation*}
			\|(N_x+\1)^2\psi\|
				\le C\bigl(\|L_x^\dagger L_x\psi\|+\|\psi\|\bigr).
		\end{equation*}
		Every Hamiltonian monomial has degree at most four and support on at most two sites. Its norm is bounded by a constant times the sum of these number weights on its support. Since the positive $L_x^\dagger L_x$ commute, $\|L_x^\dagger L_x\psi\|\le\|L'_R\psi\|$ on $\operatorname{dom}(L'_R)$. Summing $\mathcal O(|R|)$ local terms gives
		\begin{equation*}
			\|H_R^{\mathrm{cat}}\psi\|
			\le C|R|\bigl(\|L'_R\psi\|+\|\psi\|\bigr).
		\end{equation*}
		We also need bounds for Hamiltonian terms acting on code states. The code moment bounds give $\|H_R^{\mathrm{cat}}P_R\|_\infty\le C|R|$. In $L'_RH_R^{\mathrm{cat}}P_R$, the loss factor at $x$ annihilates a Hamiltonian term supported away from $x$. The remaining $\mathcal O(|R|)$ terms have uniformly bounded norm, so
		\begin{equation*}
			\|L'_RH_R^{\mathrm{cat}}P_R\|_\infty\le C|R|.
		\end{equation*}
		Finally, we apply the adiabatic estimate. The all-order Sobolev estimates and \cite[Lem.~D.2]{Rico.2025} justify the two Duhamel identities in the proof of \cite[Thm.~A.1, case $\mathcal K=0$]{Moebus.2025} for this generated dynamics. Its gap-dependent estimate, with the three bounds above, gives Equation~\eqref{eq:cat-finite-adiabatic}.
	\end{proof}

\section{Failure of uniform locality for free hopping}\label{app:free-hopping-obstruction}
	For the free-hopping example in Figure~\ref{fig:free-hopping-counterexample}, based on \cite[Supplemental Material, Section~I]{FaupinLemmSigalZhang.2025}, we prove that the local approximation error is one for every $0<t\le1$ and every spatial cut between $0$ and $r$, even as the buffer radius grows.

	Let $r\ge1$ be an integer and consider the finite chain $V=\{0,\ldots,r\}$ with Hamiltonian
	\begin{equation*}
		H=-\frac12\sum_{x=0}^{r-1}(a_x^\dagger a_{x+1}+a_{x+1}^\dagger a_x).
	\end{equation*}
	Let $\ket{\Omega}$ be the vacuum and $\ket{\psi_n}=(a_0^\dagger)^n\ket{\Omega}/\sqrt{n!}$. Let $h$ be the one-particle hopping matrix, with $h_{xy}=-1/2$ when $|x-y|=1$ and $h_{xy}=0$ otherwise. On each particle-number sector, $H$ is the sum of $h$ acting on each particle. The second-quantization identities in \cite[Thms.~13(1) and~18(3)]{Derezinski.2006} therefore give, on finite-Fock vectors,
	\begin{equation*}
		e^{-itH}a_0^\dagger e^{itH}=\sum_{x=0}^r u_{x0}(t)a_x^\dagger,
		\qquad u_{x0}(t)=(e^{-ith})_{x0}.
	\end{equation*}
	Since $H\ket{\Omega}=0$ and $\bra{\Omega}a_xa_y^\dagger\ket{\Omega}=\delta_{xy}$, these coefficients satisfy $u_{x0}(t)=\bra{\Omega}a_xe^{-itH}a_0^\dagger\ket{\Omega}$. Hence
	\begin{equation*}
		e^{-itH}\ket{\psi_n}
		=\frac{1}{\sqrt{n!}}\left(\sum_{x=0}^r u_{x0}(t)a_x^\dagger\right)^n\ket{\Omega}.
	\end{equation*}
	Thus the occupation at $r$ has a binomial distribution with parameters $n$ and $|u_{r0}(t)|^2$ \cite[Supplemental Material, Equation~(S1)]{FaupinLemmSigalZhang.2025}. The observable $A_r=\1-\ketbra{0}{0}_r$ detects a nonzero occupation at $r$. For any integer $1\le j\le r$, the evolution restricted to $R=\{j,\ldots,r\}$ preserves the initial vacuum on $R$. Hence
	\begin{equation*}
		\bra{\psi_n}\bigl(\tau_t(A_r)-\tau_t^R(A_r)\bigr)\ket{\psi_n}
		=1-\bigl(1-|u_{r0}(t)|^2\bigr)^n.
	\end{equation*}
	The right side is independent of the cut position $j$. To obtain a time independent of $r$, fix $0<t\le1$. In the Taylor series for $u_{r0}(t)$, the shortest walk from $0$ to $r$ contributes $(it/2)^r/r!$. Every remaining walk has length $r+2k$, and there are at most $\binom{r+2k}{k}$ such walks. Thus
	\begin{equation*}
		\left|u_{r0}(t)-\frac{(it/2)^r}{r!}\right|\le\frac{(t/2)^r}{r!}\sum_{k\ge1}\frac{(t^2/4)^k r!}{k!(r+k)!}\le\frac{(t/2)^r}{r!}\bigl(e^{t^2/4}-1\bigr)<\frac{(t/2)^r}{r!}.
	\end{equation*}
	Consequently, $u_{r0}(t)\ne0$. Choosing $n\ge\lceil|u_{r0}(t)|^{-2}\rceil$ gives the explicit lower bound
	\begin{equation*}
		\|\tau_t(A_r)-\tau_t^R(A_r)\|_\infty
		\ge\bra{\psi_n}\bigl(\tau_t(A_r)-\tau_t^R(A_r)\bigr)\ket{\psi_n}
		\ge1-e^{-1}.
	\end{equation*}
	Letting $n\to\infty$ gives $\|\tau_t(A_r)-\tau_t^R(A_r)\|_\infty=1$, since both evolved observables are positive contractions. Here $R=\{r\}[r-j]$, so no operator-norm local approximation bound can decay with the buffer radius uniformly over the initial occupation and finite volume.

	This is a spatial instance of a known channel obstruction. With the other sites initially in the vacuum, the channel from $0$ to $r$ is a pure-loss attenuator, up to a phase rotation. Any such channel with nonzero transmissivity has diamond distance two from vacuum replacement, as shown in \cite[Prop.~1]{Winter.2017}.

\par
\end{document}